\documentclass[11pt]{article}
\usepackage[a4paper,margin=2.6cm]{geometry}
\usepackage{amsmath,amssymb,amsthm}
\usepackage{float}
\usepackage{array}
\usepackage{algorithm,algpseudocode}
\newcommand{\algheader}[2]{\bigskip\par\noindent\textbf{#1}\ \ #2\par\nobreak\vspace{2pt}\hrule height 0.4pt\nobreak\vspace{5pt}\par\nobreak}
\newcommand{\algfooter}{\par\vspace{3pt}\hrule height 0.4pt\bigskip}
\usepackage{url}
\usepackage{hyperref}

\newtheorem{theorem}{Theorem}[section]
\newtheorem{lemma}[theorem]{Lemma}
\newtheorem{proposition}[theorem]{Proposition}
\newtheorem{corollary}[theorem]{Corollary}

\theoremstyle{definition}
\newtheorem{definition}[theorem]{Definition}
\newtheorem{example}[theorem]{Example}
\theoremstyle{remark}
\newtheorem{remark}[theorem]{Remark}

\newcommand{\Q}{\mathbb Q}
\newcommand{\Z}{\mathbb Z}

\newcommand{\cA}{\mathcal A}
\newcommand{\cS}{\mathcal S}
\newcommand{\cU}{\mathcal U}
\newcommand{\cC}{\mathcal C}
\newcommand{\cK}{\mathcal K}
\newcommand{\cE}{\mathcal E}
\newcommand{\fE}{\mathfrak E}
\newcommand{\fp}{\mathfrak p}
\newcommand{\Fbar}{\overline{F}}
\newcommand{\den}{\operatorname{den}}
\newcommand{\Const}{\operatorname{Const}}

\newcommand{\Ei}{\operatorname{Ei}}
\newcommand{\li}{\operatorname{li}}
\newcommand{\erf}{\operatorname{erf}}
\newcommand{\erfc}{\operatorname{erfc}}
\newcommand{\erfi}{\operatorname{erfi}}
\newcommand{\Si}{\operatorname{Si}}
\newcommand{\Ci}{\operatorname{Ci}}
\newcommand{\Li}{\operatorname{Li}}
\newcommand{\res}{\operatorname{res}}
\newcommand{\GQ}{\mathcal{G}_{\Q}}

\title{Parallel Integration over Simple Radical Extensions III:\\ Antiderivatives in Terms of Special Functions}
\author{Sam Blake}
\date{\today}

\begin{document}
\maketitle

\begin{abstract}
We extend the parallel integration method for mixed towers of Part~II from elementary antiderivatives to antiderivatives in special functions. The class covered is the incomplete gamma function $\Gamma(s,\cdot)$ at rational $s$, which contains $\Ei$, $\li$, $\Si$, $\Ci$, $\erf$ and the Fresnel integrals, together with the elliptic integrals $F$, $E$, $\Pi$. Each special function enters through a \emph{kernel}: a known element of the tower whose antiderivative is that function. It is either fixed by residues or added as one more column of the single linear system, and no Risch differential equation is solved. We prove where kernels can have poles; the $\erf$ kernels live in the sub-critical window of radical towers. The denominator theory, degree bounds and certificates of Part~II carry over, and new criteria decide most places that Part~II leaves to a guess. Elliptic integrals are carried by the radical, and a non-torsion residue divisor becomes a third-kind term. We prove that special functions are introduced only when necessary. Elementary answers are returned unchanged, and in strict mode a special function comes with a certificate that the integrand has no elementary integral. When no complete answer is found, a partial answer with a reduced remainder is returned. All examples are computed by a SymPy implementation and verified by differentiation.
\end{abstract}

\section{Introduction}

A parallel integrator guesses the shape of an antiderivative and solves one linear system for its coefficients \cite{NormanMoore77,Davenport82,DavenportTrager85,Bronstein07}, \cite[Ch.~10]{Bronstein05}. Parts~I and II of this series \cite{PartI,PartII} proved the structural facts behind the guess for towers containing one simple radical, and turned the method into an algorithm that returns an elementary antiderivative, certifies non-elementarity, or reports failure honestly. When the antiderivative is not elementary it is often expressible in named special functions,
\[
\int e^{x^2}\,dx=\tfrac{\sqrt\pi}2\erfi(x),\qquad
\int\sqrt{\log x}\,dx=x\sqrt{\log x}-\tfrac{\sqrt\pi}2\erfi\bigl(\sqrt{\log x}\bigr),
\]
\[
\int\frac{dx}{\sqrt{e^{x}-e^{-x}}}=-F\bigl(\varphi\,\big|\,\tfrac12\bigr),\quad \sin\varphi=\sqrt{\tfrac{2}{e^x+1}},
\]
and a user of a computer algebra system expects to see these answers. This paper extends the parallel method of Part~II to them.

Liouville's theorem has been extended to special functions in two directions. Singer, Saunders and Caviness \cite{SSC85} allowed error functions and logarithmic integrals and showed that they enter the antiderivative linearly with constant coefficients; Cherry \cite{Cherry85,Cherry86} and Knowles \cite{Knowles92,Knowles93} turned this into recursive decision procedures for transcendental towers; Hebisch \cite{Hebisch15,Hebisch18} extended the principle to exponential integrals and incomplete gamma functions at rational parameter, with an implementation in FriCAS, and to polylogarithms \cite{Hebisch18b}; Baddoura \cite{Baddoura06} and Kaur and Srinivasan \cite{KaurSrinivasan19} treated dilogarithms. For elliptic integrals, Hebisch \cite{Hebisch21} proved the Liouville principle in the form we use, and Kumbhakar and Srinivasan \cite{KumbhakarSrinivasan23} showed that adjoining non-Liouvillian functions such as complete elliptic integrals of a variable modulus or Weierstrass functions never helps. On the parallel side, Bronstein's \texttt{pmint} \cite{pmint} and Norman's completion method \cite{Norman90} admit special functions \emph{in the integrand} as generators of the tower, and Raab \cite{Raab12,Raab13} extended the recursive algorithm to parametric problems. None of these works combines special functions in the answer with the parallel method, and none treats algebraic extensions at arbitrary positions of the tower.

The functions supported, and their precise definitions, are listed in Tables~\ref{tab:answers} and \ref{tab:integrands} (\S\ref{sec:functions}). The contributions are as follows.
\begin{enumerate}
\item \emph{Kernels} (Definition~\ref{def:kernel}): one first-order condition $D\omega/\omega=Dv+(s-1)Dv/v$ describes every term of the Liouville principle for $\Gamma(s,\cdot)$, including those whose factors $(-v)^{s-1}$ and $e^{v}$ do not lie in the tower while their product does (\S\ref{ex:hebisch22}).
\item \emph{Kernels at normal primes} (Proposition~\ref{prop:kernelnormal}): the $s=0$ kernels are critical with residue $v_P(v)\,\bar\xi$; the $s>0$ kernels are never critical and have poles only where $\delta_P\ge2$; the argument's zero multiplicities constrain $s$, and ramification multiplies them (Corollary~\ref{cor:parameter}).
\item \emph{Structure} (Theorem~\ref{thm:structureG}): Part~II's denominator bound and degree bounds hold with the integrand's valuation replaced by a minimum over the kernels, and the canonical residue decomposes into a logarithmic and an exponential part; a residue outside $\Fbar+\sum\Fbar\bar\xi$ certifies that no antiderivative exists in the class (Corollary~\ref{cor:rescert}).
\item \emph{Bounds} (Section~\ref{sec:bounds}): Part~II's cancellation theorem extended to residue fields with extra constants (Theorem~\ref{thm:extra}), three new criteria and the branch places over special primes (Proposition~\ref{prop:K567}), and the identification of the remaining heuristic places with the zeros of exponential elements (Proposition~\ref{prop:kerexp}); on the examples, every bound is proved in $14$ of $18$ cases instead of $8$.
\item \emph{Candidates} (Section~\ref{sec:candidates}): the $\Ei$ terms are determined by residues or by the values of logarithms of exponential elements at special places; the $\Gamma(s,\cdot)$ terms with $s>0$ by a finite Kummer search.
\item \emph{Elliptic integrals} (Section~\ref{sec:elliptic}): points of constant elliptic curves over the tower live on the radical (Proposition~\ref{prop:points}); $2g$ columns suffice for the first and second kinds (Proposition~\ref{prop:columns}); residues pull back as constants (Lemma~\ref{lem:respull}); the odd residues at a simple pole are removed by an explicit third-kind kernel (Proposition~\ref{prop:third}); and Legendre forms are derived for the presentation (Lemma~\ref{lem:legendre}).
\item \emph{Partial answers} (Section~\ref{sec:partial}): when no antiderivative is found, $\int f=I+\int r$ with $I$ in the class and a reduced remainder $r$ --- a normal form modulo the ansatz, with the Hermite part and every integrable residue removed --- and a certificate for $f$ is one for $\int r$ (Theorem~\ref{thm:partial}).
\item \emph{The algorithm} (Section~\ref{sec:algorithm}) and its SymPy implementation, which computes every example of Section~\ref{sec:examples} and verifies it by differentiation. It returns Part~II's answer unchanged whenever Part~II finds one, and in strict mode it returns a special function only when the integrand is certified to have no elementary integral (Theorem~\ref{thm:necessity}); an independence test proves every special term forced (Theorem~\ref{thm:independence}).
\end{enumerate}

The paper keeps the lines drawn in the series: one ring, one ansatz, one linear system; no recursion on the tower, no Risch differential equation, no integration by parts. A special function is \emph{parametric}: its kernel is a known element of the tower, and its coefficient is an unknown constant of the same linear system as the numerator of the rational part.

\section{Setting and Liouville principles}\label{sec:setting}

Throughout, $(L,D)$ is a tower satisfying (T1)--(T4) of \cite[\S2]{PartII}: $K_0=F(t_1,\dots,t_j)$ closed under $D$, a radical $y^m=q$ with $q\in R_0=F[t_1,\dots,t_j]$ normalised as in Part~I, upper generators $t_{j+1},\dots,t_n$ whose derivatives may involve $y$, and $\Const_D(L)=F$. We use the ring $\cA$, the valuations $v_P$ at height-one primes, the pole orders $\eta_P$ and shifts $\delta_P=1+\eta_P$, the normal/special dichotomy, the valuation lemma \cite[Thm.~4.5]{PartII}, the logarithmic-derivative lemma \cite[Lemma~5.3]{PartII}, the residues $\tau_P$ \cite[Thm.~7.4]{PartII} and the canonical residue $\hat\tau_P$ \cite[Prop.~7.6]{PartII}. $\Fbar$ is an algebraic closure of $F$.

\subsection{Kernels}

\begin{definition}[Kernel data]\label{def:kernel}
Let $s\in\Q\cap[0,1)$ and let $M\supseteq L$ be a differential field. A pair $(v,\omega)\in M\times M^*$ with $v\notin\Const(M)$ is an \emph{$s$-kernel datum} over $M$ if
\begin{equation}\label{eq:kernel}
\frac{D\omega}{\omega}\;=\;Dv+(s-1)\frac{Dv}{v}.
\end{equation}
Its \emph{kernel} is $K_s(v,\omega):=\omega\,Dv$. For $s=0$ we write $\xi:=v\omega$, so that $D\xi/\xi=Dv$ and $K_0=\xi\,Dv/v$.
\end{definition}

\begin{lemma}[Analytic meaning]\label{lem:analytic}
If $(v,\omega)$ is an $s$-kernel datum then $\omega=C\,(-v)^{s-1}e^{v}$ for a constant $C$ of a differential extension, and $\int K_s(v,\omega)=C\,\Gamma(s,-v)$. For $s=0$, $\int\xi\,Dv/v=\kappa\,\Ei(v)$ with the constant $\kappa=\xi e^{-v}$. Both hold up to additive constants that depend on the branch.
\end{lemma}

\begin{proof}
By \eqref{eq:kernel}, $D\bigl(\omega(-v)^{1-s}e^{-v}\bigr)=0$. Since $\partial_z\Gamma(s,z)=-z^{s-1}e^{-z}$, $D\,\Gamma(s,-v)=(-v)^{s-1}e^{v}Dv$. For $s=0$, $\Gamma(0,-v)=E_1(-v)=-\Ei(v)$ modulo a branch constant.
\end{proof}

The rational gamma extensions of Hebisch \cite[\S4]{Hebisch18} adjoin exponentials, logarithms, algebraic elements and integrals of the form $\theta'=v'u/v$ with $u'=v'u$ (exponential integrals) and $\theta'=wu'$ with $u'=v'u$, $w^k=v^{m}$, $-k<m<0$ (incomplete gamma functions). The first is $K_0$ with $\xi=u$, the second is $K_s$ with $s=(k+m)/k$ and $\omega=wu$ up to a constant. We write $\GQ$ for the class of functions in rational gamma extensions of $L$. The reformulation matters for radical towers: \eqref{eq:kernel} involves $\omega$ only through its logarithmic derivative, so the condition ``$wu\in L$'' of \cite[Thm.~18]{Hebisch18} becomes $\omega\in L$ --- one element, even when neither $w$ nor $u$ lies in $L$.

\begin{theorem}[{Liouville principle for $\GQ$; \cite[Thm.~1.1]{SSC85}, \cite[Thm.~13]{Hebisch18}}]\label{thm:liouville}
If $f\in L$ has an antiderivative in a rational gamma extension of $L$ with constant field $\Fbar$, then there are a finite algebraic extension $L'$ of $\Fbar L$, $v_0\in L'$, $c_1,\dots,c_r\in\Fbar$, $u_1,\dots,u_r\in L'^*$ and $s_l$-kernel data $(v_l,\omega_l)$ over $L'$ with $\gamma_l\in\Fbar$ such that
\begin{equation}\label{eq:liouG}
f\;=\;Dv_0+\sum_{i=1}^rc_i\frac{Du_i}{u_i}+\sum_{l=1}^k\gamma_l\,\omega_l\,Dv_l .
\end{equation}
\end{theorem}

For rational gamma extensions no log-explicitness hypothesis is needed \cite[Thm.~13]{Hebisch18}. The special functions enter linearly with constant coefficients, as in \cite{SSC85}: this is what makes them columns of a linear system. For the logarithmic part the data may be taken in $\Fbar L$ \cite[eq.~(3)]{PartII}; for the kernels this is a hypothesis in the present paper.

\begin{enumerate}
\item[(H1)] \emph{Descent.} In \eqref{eq:liouG} one may take $v_0,u_i,v_l,\omega_l\in\Fbar L$.
\end{enumerate}

Hebisch proves the descent through transcendental monomial layers \cite[Lemmas~15, 17, Thm.~18]{Hebisch18}; the radical layer is the open case, and his Example~22 (a variant of \cite[Ex.~1.2]{SSC85}) shows the pitfall: there the integrand lies in a radical extension and its trace to the field below vanishes, so the descent cannot pass through a trace of $f$. (H1) is used only for the certificate of Corollary~\ref{cor:rescert}; every antiderivative the algorithm returns is verified, and nothing else depends on it.

\subsection{Elliptic integrals}

\begin{theorem}[{Liouville principle for elliptic integrals; Hebisch \cite[Thm.~4]{Hebisch21}}]\label{thm:ellliouville}
Let $f\in L$ have a primitive in an extension of $L$ built from elementary extensions, elliptic functions, elliptic integrals and Lambert functions. Then, after an extension of $L$ by algebraic constants,
\[
f=Dv_0+\sum_ic_i\frac{Dv_i}{v_i}+\sum_ja_j\frac{Dp_j}{q_j}+\sum_jb_j\frac{p_jDp_j}{q_j}+\sum_jc'_j\frac{Dp_j}{(p_j-\gamma_j)q_j}
\]
with $v_0,v_i,p_j,q_j$ in $L$ extended by algebraic constants, constants $c_i,a_j,b_j,c'_j,\gamma_j$, and $q_j^2=p_j^3-\alpha_jp_j-\beta_j$ with $\alpha_j,\beta_j$ constant.
\end{theorem}

Each elliptic term is the pull-back $P_j^*\omega$ of a differential of the first, second or third kind on a \emph{constant} elliptic curve $\cE_j$ by a point $P_j=(p_j,q_j)\in\cE_j(\Fbar L)$. Here the descent is part of the statement: Hebisch's trace argument uses Abel's addition theorem to keep the points in the field.

\begin{theorem}[{Kumbhakar--Srinivasan \cite[Thm.~1.1]{KumbhakarSrinivasan23}; Hebisch \cite[\S5]{Hebisch21}}]\label{thm:negative}
Let $E\supseteq L$ be a tower with the same constants in which each step is elementary, or a Picard--Vessiot extension with differential Galois group $\mathrm{SL}_2$ or $\mathrm D_\infty$, or the adjunction of a Weierstrassian element. If $f\in L$ has an elementary integral over $E$, it has one over $L$ \textup(constants extended to $\Fbar$\textup).
\end{theorem}

Complete elliptic integrals $K(m)$, $E(m)$ of a non-constant modulus $m\in L$ generate a Picard--Vessiot extension for Legendre's equation, whose group is $\mathrm{SL}_2$ when they are not Liouvillian over $L$; theta constants of a variable nome are expressed through them; Weierstrass functions of a new argument are Weierstrassian. By Theorem~\ref{thm:negative} none of them occurs non-trivially in an otherwise elementary antiderivative of an element of $L$, and Hebisch sketches the same conclusion with elliptic integrals in the tower \cite[\S5]{Hebisch21}. They appear in our answers only in two roles: because they are in the integrand, where they are generators of a Part~II tower (\S\S\ref{ex:ellK}, \ref{ex:sn2}); or, with constant modulus, as constants of the Legendre and Jacobi presentations (Section~\ref{sec:presentation}).

\subsection{The functions supported, and their definitions}\label{sec:functions}

Special functions are normalised differently by different computer algebra systems and references: modulus versus parameter, the sign in $\Pi$, $\li$ versus the offset integral $\li(z)-\li(2)$ (which some references call $\Li$ --- while Maple's \texttt{Li} is $\li$ itself), and the branch of $\Ei$. We therefore fix the definitions. Table~\ref{tab:answers} lists every special function that the algorithm returns in an answer, and Table~\ref{tab:integrands} every function that it accepts in an integrand. Nothing outside these tables is produced.

In Table~\ref{tab:answers}, $\log$, $\sqrt{\cdot}$ and $\arcsin$ are principal branches, $\gamma$ is Euler's constant, and every function is continued analytically from the stated real interval. The names are those of SymPy, used by the implementation, of Mathematica and of Maple, in that order.

\begin{table}[ht]
\centering\small
\renewcommand{\arraystretch}{1.35}
\begin{tabular}{|p{1.9cm}|>{\raggedright\arraybackslash}p{5.2cm}|>{\raggedright\arraybackslash}p{4.3cm}|>{\raggedright\arraybackslash}p{2.2cm}|}
\hline
Function & Definition & SymPy\newline Mathematica\newline Maple & Produced by \\
\hline
$\Gamma(s,z)$, $s\in\Q$, $0<s<1$ & $\displaystyle\int_z^\infty t^{s-1}e^{-t}\,dt$ for $z>0$; cut along $(-\infty,0]$ & \texttt{uppergamma(s,z)}\newline \texttt{Gamma[s,z]}\newline \texttt{GAMMA(s,z)} & $K_s$, $s>0$ (Lemmas~\ref{lem:analytic}, \ref{lem:kummer}) \\
\hline
$\Ei(z)$ & $\displaystyle\gamma+\log z+\sum_{k\ge1}\frac{z^k}{k\,k!}$; for $x>0$, $\mathrm{PV}\!\int_{-\infty}^x\frac{e^t}{t}\,dt$ & \texttt{Ei(z)}\newline \texttt{ExpIntegralEi[z]}\newline \texttt{Ei(z)} & $K_0$ (Lemma~\ref{lem:analytic}) \\
\hline
$\li(z)$ & $\mathrm{PV}\!\displaystyle\int_0^z\frac{dt}{\log t}=\Ei(\log z)$, $z>0$, $z\ne1$; not the offset integral $\li(z)-\li(2)$ & \texttt{li(z)}\newline \texttt{LogIntegral[z]}\newline \texttt{Li(z)} & $K_0$ with $v=\log z$ \\
\hline
$\Si(z)$ & $\displaystyle\int_0^z\frac{\sin t}{t}\,dt$ (entire) & \texttt{Si(z)}\newline \texttt{SinIntegral[z]}\newline \texttt{Si(z)} & pairs of $K_0$ (\S\ref{sec:presentation}) \\
\hline
$\Ci(z)$ & $\displaystyle\gamma+\log z+\int_0^z\frac{\cos t-1}{t}\,dt$ & \texttt{Ci(z)}\newline \texttt{CosIntegral[z]}\newline \texttt{Ci(z)} & pairs of $K_0$ \\
\hline
$\erfc(z)$ & $\displaystyle\frac2{\sqrt\pi}\int_z^\infty e^{-t^2}\,dt=1-\erf(z)$, $\ \erf(z)=\frac2{\sqrt\pi}\int_0^ze^{-t^2}dt$ & \texttt{erfc(z)}\newline \texttt{Erfc[z]}\newline \texttt{erfc(z)} & $K_{1/2}$ \\
\hline
$\erfi(z)$ & $-\sqrt{-1}\,\erf(\sqrt{-1}\,z)=\displaystyle\frac2{\sqrt\pi}\int_0^ze^{t^2}\,dt$ & \texttt{erfi(z)}\newline \texttt{Erfi[z]}\newline \texttt{erfi(z)} & $K_{1/2}$ \\
\hline
$F(\varphi\,|\,m)$ & $\displaystyle\int_0^\varphi\frac{d\theta}{\sqrt{1-m\sin^2\theta}}$ & \texttt{elliptic\_f(phi,m)}\newline \texttt{EllipticF[phi,m]}\newline \texttt{EllipticF}$(\sin\varphi,\sqrt m)$ & $K^{(0)}$, $K^{(1)}$, $K_p$ \\
\hline
$E(\varphi\,|\,m)$ & $\displaystyle\int_0^\varphi\sqrt{1-m\sin^2\theta}\,d\theta$ & \texttt{elliptic\_e(phi,m)}\newline \texttt{EllipticE[phi,m]}\newline \texttt{EllipticE}$(\sin\varphi,\sqrt m)$ & $K^{(1)}$; Legendre $K^{(2)}$ \\
\hline
$\Pi(n;\varphi\,|\,m)$ & $\displaystyle\int_0^\varphi\frac{d\theta}{(1-n\sin^2\theta)\sqrt{1-m\sin^2\theta}}$ & \texttt{elliptic\_pi(n,phi,m)}\newline \texttt{EllipticPi[n,phi,m]}\newline \texttt{EllipticPi}$(\sin\varphi,n,\sqrt m)$ & $K_p$ (Prop.~\ref{prop:third}) \\
\hline
\end{tabular}
\caption{The special functions returned in answers.}\label{tab:answers}
\end{table}

\emph{Conventions.} Maple's names need two warnings. First, Maple writes the logarithmic integral as \texttt{Li}: its \texttt{Li(z)} is our $\li(z)$, not the offset integral. Second, Maple's elliptic integrals take $\sin\varphi$ and the modulus $k=\sqrt m$ as arguments: $F(\varphi\,|\,m)=\texttt{EllipticF}(\sin\varphi,\sqrt m)$, and likewise for $E$ and $\Pi$, with the characteristic $n$ as the middle argument of \texttt{EllipticPi}. These conversions hold for $|\operatorname{Re}\varphi|\le\pi/2$, where every amplitude produced by Lemma~\ref{lem:legendre} lies. The complete integrals of Table~\ref{tab:integrands} are $K(m)=\texttt{EllipticK}(\sqrt m)$ and $E(m)=\texttt{EllipticE}(\sqrt m)$ in Maple, and \texttt{EllipticK[m]}, \texttt{EllipticE[m]} in Mathematica. The elliptic integrals use the \emph{parameter} $m$. References that use the modulus $k$ write $F(\varphi,k)=F(\varphi\,|\,k^2)$, and some write $\Pi$ with $1+n\sin^2\theta$, which is $\Pi(-n;\varphi\,|\,m)$ here; these are the conventions of \cite[Ch.~19]{DLMF} in the $(\varphi\,|\,m)$ notation. The amplitude produced by Lemma~\ref{lem:legendre} is $\varphi=\arcsin\bigl(\lambda/\sqrt{U-e_3}\bigr)$, principal branch.

\emph{Identities used by the presentation.} The following identities, checked numerically to $25$ digits, are the only rewritings the implementation performs:
\begin{itemize}
\item $\Gamma(0,z)=E_1(z)=-\Ei(-z)$ for $z>0$;
\item $\Gamma(\tfrac12,z^2)=\sqrt\pi\,\erfc(z)$ for $\operatorname{Re}z>0$;
\item $\erfc(\sqrt{-1}\,z)=1-\sqrt{-1}\,\erfi(z)$;
\item $\li(z)=\Ei(\log z)$;
\item $\Ei(\pm\sqrt{-1}\,x)=\Ci(x)\pm\sqrt{-1}\bigl(\Si(x)+\tfrac\pi2\bigr)$ for $x>0$.
\end{itemize}
Since each special term of an antiderivative has a constant coefficient, these are used up to additive constants. The constants may differ on different components of the domain, which is why every answer is verified by differentiation rather than by comparing values.

\emph{Described but not produced.} Several functions appear in the theory but are never returned by the implementation:
\begin{itemize}
\item $\operatorname{Shi}$ and $\operatorname{Chi}$: $\int\sinh(x)/x\,dx$ is returned as $\tfrac12(\Ei(x)-\Ei(-x))$;
\item the Fresnel integrals $S$ and $C$: $\int\sin(x^2)\,dx$ is returned through $\erfc$ at the arguments $\sqrt{\pm\sqrt{-1}}\,x$;
\item the Weierstrass functions $\wp^{-1}$, $\zeta$, $\sigma$, and Jacobi's $Z$ and $\Theta$ with the constants $K(m)$, $E(m)$ of \S\ref{sec:presentation};
\item $\Gamma(s,\cdot)$ at irrational $s$;
\item integrals on pencils of genus $\ge2$, and elliptic integrals on pencils other than cubics and Legendre quartics, which are returned as unevaluated integrals in the pencil variable.
\end{itemize}

\begin{table}[ht]
\centering\small
\renewcommand{\arraystretch}{1.3}
\begin{tabular}{|p{4.6cm}|p{9.6cm}|}
\hline
Integrand function & How it enters the tower \\
\hline
$\exp$, $\log$, $a^b=\exp(b\log a)$ & hyperexponential and primitive generators; a new $\exp(a)$ or $\log(a)$ is first tested against the existing ones by the structure theorem (below) \\
\hline
$y=q^{1/m}$, rational powers & one simple radical, with the flattening and Euler steps of Part~II \\
\hline
$\sin,\cos,\tan,\cot,\sec,\csc$ and the hyperbolic six & the hypertangent (resp.\ hyperbolic) half-angle generator $\tan(a/2)$ (resp.\ $\tanh(a/2)$) \\
\hline
$\arctan$, $\operatorname{arccot}$, $\operatorname{artanh}$, $\operatorname{arcoth}$; $\arcsin$, $\arccos$, $\operatorname{arsinh}$, $\operatorname{arcosh}$, $\operatorname{arcsec}$, $\operatorname{arccsc}$, $\operatorname{arsech}$, $\operatorname{arcsch}$ & primitive generators, with a radical for the second group \\
\hline
$W(z)$ (Lambert, principal branch $W_0$, $We^W=z$) & non-monomial generator, $DW=W\,Dz/(z(1+W))$; the source $(z/W,W)$ \\
\hline
\emph{tower level only:} $K(m)$, $E(m)$; $\operatorname{sn},\operatorname{cn},\operatorname{dn}$; $\wp$, $\wp'$; any generator satisfying (T3) & supplied by the user as a \texttt{Tower}; $DK=(E-(1-m)K)/(2m(1-m))$, $DE=(E-K)/(2m)$; $D\operatorname{sn}=\operatorname{cn}\operatorname{dn}$ with the rescaling of Part~II Lemma~3.4 (\S\ref{ex:ellK}) \\
\hline
\end{tabular}
\caption{The functions accepted in integrands.}\label{tab:integrands}
\end{table}

\emph{Dependent exponentials and logarithms.}\label{par:structure} The surface builder must keep the generators algebraically independent, hypothesis (T1) of Part~II. Otherwise $e^{2x}$ becomes a generator independent of $e^x$, and false certificates follow, e.g.\ for $\int e^{2x}/(e^x+1)\,dx=e^x-\log(1+e^x)$. Before adjoining $e^a$ or $\log a$ it therefore applies the Risch structure theorem \cite{RothsteinCaviness79}. It solves $Da=\sum r_kD\theta_k/\theta_k+\sum r_kDt_k$, respectively $Da/a$ equal to the same sum, for rational $r_k$ over the exponential generators $\theta_k=e^{\eta_k}$ and the logarithmic generators $t_k=\log u_k$ already in the tower. If such $r_k$ exist:
\begin{itemize}
\item a new $\log a$ is replaced by $\sum r_k(\eta_k\text{ or }t_k)+\log c$;
\item a new $e^a$ with integer exponents is replaced by $e^c\prod(\theta_k\text{ or }u_k)^{r_k}$;
\item a fractional exponent $r_k=p/d$ refines the generator, $\theta_k=\theta'^d$ with $\theta'=e^{\eta_k/d}$.
\end{itemize}
This is a correction to Part~II's builder. On the integrands of our tests it changes Part~II's answer in fourteen cases, all with dependent exponentials or logarithms, and seven of the old answers were false certificates of non-elementarity.

\section{Kernels at normal primes}\label{sec:normal}

\begin{proposition}[Kernels at a normal prime]\label{prop:kernelnormal}
Let $P$ be a normal height-one prime of $\Fbar\cA$ \textup(or the hypertangent place at infinity \cite[Lemma~8.1(c)]{PartII}\textup), and let $(v,\omega)$ be an $s$-kernel datum over $\Fbar L$. Then:
\begin{enumerate}
\item[(a)] $v_P(v)\ge0$;
\item[(b)] with $n:=v_P(v)$, one has $v_P(\omega)=(s-1)\,n$; in particular $(1-s)\,n\in\Z$;
\item[(c)] $v_P\bigl(K_s(v,\omega)\bigr)=s\,n-\delta_P$ if $n>0$, and $v_P\bigl(K_s(v,\omega)\bigr)\ge1-\delta_P$ if $n=0$;
\item[(d)] if $s=0$, then $\xi=v\omega$ is a $P$-unit and $\hat\tau_P(K_0)=n\,\bar\xi$, where $\bar\xi$ is the reduction of $\xi$ in $\kappa(P)$;
\item[(e)] if $s>0$, then $v_P(K_s)\ge1-\delta_P$ and $\hat\tau_P(K_s)=0$: a kernel with $s>0$ has a pole at $P$ only if $\delta_P\ge2$, and then of order at most $\delta_P-1$.
\end{enumerate}
\end{proposition}

\begin{proof}
(a) Rewrite \eqref{eq:kernel} as $Dv=D\omega/\omega-(s-1)Dv/v$. By \cite[Lemma~5.3]{PartII} both logarithmic derivatives on the right have valuation $\ge-\delta_P$. If $v_P(v)<0$, the valuation lemma gives $v_P(Dv)=v_P(v)-\delta_P<-\delta_P$, a contradiction.

(b) Let $\pi$ be a uniformiser. By \cite[Lemma~5.3]{PartII},
\[
\frac{D\omega}{\omega}=v_P(\omega)\frac{D\pi}{\pi}+\varepsilon_1,\qquad \frac{Dv}{v}=n\,\frac{D\pi}{\pi}+\varepsilon_2,\qquad v_P(\varepsilon_i)\ge1-\delta_P,
\]
and $v_P(Dv)\ge1-\delta_P$ by (a) and the valuation lemma (for $n>0$ directly, for $n=0$ by its last clause). Substituting into \eqref{eq:kernel},
\[
\bigl(v_P(\omega)-(s-1)n\bigr)\frac{D\pi}{\pi}\;\in\;\pi^{1-\delta_P}\cA_P .
\]
Since $v_P(D\pi/\pi)=-\delta_P$ exactly and the coefficient is a rational number, it vanishes.

(c) $v_P(K_s)=v_P(\omega)+v_P(Dv)$, which is $(s-1)n+(n-\delta_P)=sn-\delta_P$ for $n>0$ and $\ge0+(1-\delta_P)$ for $n=0$.

(d) For $s=0$, $v_P(\xi)=n-n=0$. If $n=0$ then $v_P(K_0)\ge1-\delta_P$ and $\hat\tau_P(K_0)=0$. If $n>0$, $v_P(K_0)=-\delta_P$, and \cite[Thm.~7.4(i)]{PartII} with $h=v$, $e=n$ gives $\tau_P(K_0)=n\,(K_0\,v/Dv)|_P=n\,\bar\xi$; the canonical residue equals $\tau_P$ at valuation $\ge-\delta_P$ \cite[Prop.~7.6]{PartII}.

(e) For $n>0$, $sn=n-(1-s)n$ is an integer by (b) and is positive, so $sn\ge1$ and (c) gives $v_P(K_s)\ge1-\delta_P$; for $n=0$ this is (c). A function of valuation $>-\delta_P$ has canonical residue zero.
\end{proof}

\begin{corollary}[Parameter constraint]\label{cor:parameter}
Let $s=a/k$ in lowest terms, $s>0$, and let $(v,\omega)$ be an $s$-kernel datum over $\Fbar L$. Then $k$ divides $v_P(v)$ at every normal prime $P$ at which $v$ vanishes. For $v\in\Fbar K_0$ and $P$ over the irreducible $p\in R_0$ with ramification index $e_P$, $v_P(v)=e_P\operatorname{ord}_p(v)$. Hence in a transcendental tower every normal zero of the argument of $\Gamma(a/k,\cdot)$ has multiplicity divisible by $k$, while at a branch prime of the radical a simple zero of $v$ in $K_0$ admits every $k\mid e_P$.
\end{corollary}

\begin{proof}
$(1-s)n=(k-a)n/k\in\Z$ with $\gcd(k,k-a)=1$ gives $k\mid n$. The formula for $v_P$ on $K_0$ is \cite[Lemma~4.7]{PartII}.
\end{proof}

\begin{remark}[The window is where $\erf$ lives]\label{rem:window}
Part~II isolated, among the phenomena without transcendental analogue, a \emph{sub-critical window}: at a normal prime with $\delta_P\ge2$, poles of order $<\delta_P$ demand neither a Hermite denominator nor a residue. Proposition~\ref{prop:kernelnormal}(e) identifies its non-elementary content: the kernels of $\Gamma(s,\cdot)$, $s>0$, have poles at normal primes only there. In the monomial transcendental towers of \cite[Ch.~10]{Bronstein05} every normal prime has $\delta_P=1$ \cite[Rem.~4.11]{PartII}, so $\erf$-kernels have no normal poles at all and are visible only at special places --- the classical reason that the error function is an RDE phenomenon. In the towers of Part~II the window is supplied by the branch primes of the radical, where $\delta_P=e_P$, and by the primes of $\den_0$ of non-monomial generators, such as $(t)$ for $t=\sqrt{\log x}$, where $\delta=2$; the kernels of \S\S\ref{ex:sqrtlog}, \ref{ex:conic}, \ref{ex:hebisch22} sit there.
\end{remark}

\begin{remark}[Locality]
A kernel with $s>0$ is \emph{locally exact} at every normal prime: a function of valuation $a-\delta_P$ with $a\ge1$ is congruent to $D(c\pi^{a})$ for a suitable $P$-integral $c$ \cite[proof of Prop.~7.6]{PartII}. Its non-elementarity is therefore global, and its candidates cannot be read off residues; they are generated structurally in Section~\ref{sec:kummer}. For $s=0$ the opposite holds: the kernel is critical, and its residue is the handle by which the $\Ei$ terms are found.
\end{remark}

\section{Structure of antiderivatives in $\GQ$}\label{sec:structure}

\begin{definition}[Exponential elements]\label{def:expgroup}
An element $\xi\in(\Fbar L)^*$ is \emph{exponential} if $D\xi/\xi\in D(\Fbar L)$. The exponential elements form a group $\fE\supseteq\Fbar^*$, and $\lambda:\fE\to\Fbar L/\Fbar$, $D\lambda(\xi)=D\xi/\xi$, is the \emph{logarithm map}. For a normal prime $P$, $\fE_P$ denotes the subgroup of $\xi$ such that $\lambda(\xi)$ has a $P$-integral representative whose reduction modulo $P$ is a constant $c_P(\xi)\in\Fbar$.
\end{definition}

Exponential elements are $P$-units at normal primes: by the argument of Proposition~\ref{prop:kernelnormal}(a), $\lambda(\xi)$ is $P$-integral, so $v_P(D\xi/\xi)\ge1-\delta_P$ and $v_P(\xi)=0$ by \cite[Lemma~5.3]{PartII}. In an $s=0$ datum, $\xi=v\omega\in\fE$ and $v=\lambda(\xi)$ up to a constant: the arguments of the $\Ei$ terms are logarithms of exponential elements, shifted by constants.

\begin{theorem}[Structure with kernels]\label{thm:structureG}
Let $f\in L$ satisfy \eqref{eq:liouG} with $v_0,u_i,v_l,\omega_l\in\Fbar L$ and $c_1,\dots,c_r$ linearly independent over $\Q$ \cite[Lemma~5.2]{PartII}. Then:
\begin{enumerate}
\item[(i)] at every normal prime $P$, $v_P(v_0)\ge\min\bigl(0,\ \min(v_P(f),\min_lv_P(K_l))+\delta_P\bigr)$; in particular the kernels with $s=0$ and the kernels with $s>0$ impose no Hermite denominator, and \cite[Cor.~5.6]{PartII} holds verbatim: $\den(v_0)\mid s\prod_{l\ge2}d_l^{\,l-1}$ with $s$ supported on primes carrying a special prime;
\item[(ii)] at every normal prime $P$,
\begin{equation}\label{eq:resdecomp}
\hat\tau_P(f)\;=\;\sum_{i=1}^rc_i\,v_P(u_i)\;+\;\sum_{l:\,s_l=0}\gamma_l\,v_P(v_l)\,\bar\xi_l\ \in\ \Fbar+\sum_{\xi\in\fE_P}\Fbar\,\bar\xi\ \subseteq\ \kappa(P).
\end{equation}
\end{enumerate}
\end{theorem}

\begin{proof}
(i) If $v_P(v_0)<0$, the valuation lemma gives $v_P(Dv_0)=v_P(v_0)-\delta_P<-\delta_P$, while the logarithmic derivatives have valuation $\ge-\delta_P$ and the kernels valuation $\ge-\delta_P$ (Proposition~\ref{prop:kernelnormal}(c)); so $v_P(f-\sum_l\gamma_lK_l)=v_P(v_0)-\delta_P$. The corollary is the proof of \cite[Cor.~5.6]{PartII} unchanged.

(ii) The canonical residue is $\Fbar$-linear, vanishes on $D(\Fbar L)$ and satisfies $\hat\tau_P(Du/u)=v_P(u)$ \cite[Prop.~7.6]{PartII}; on the kernels it is given by Proposition~\ref{prop:kernelnormal}(d),(e). An $s=0$ kernel has $v_l=\lambda(\xi_l)+c$ with $c$ constant; if $v_P(v_l)>0$ then $\lambda(\xi_l)\equiv-c$ modulo $P$, so $\xi_l\in\fE_P$.
\end{proof}

\begin{corollary}[Refined residue certificate]\label{cor:rescert}
Assume \textup{(H1)}. If at some normal prime $P$ the canonical residue $\hat\tau_P(f)$ does not lie in $\Fbar+\sum_{\xi\in\fE_P}\Fbar\,\bar\xi$, then $f$ has no antiderivative in $\GQ$. If it lies there with a nonzero component along some $\bar\xi$ that is independent of $1$ and of the other reductions, every antiderivative of $f$ in $\GQ$ contains an $s=0$ kernel whose argument $\lambda(\xi)-c_P(\xi)$ vanishes at $P$.
\end{corollary}

This refines \cite[Prop.~9.2(a)]{PartII}: there a non-constant residue certified non-elementarity; here it certifies non-membership in $\GQ$ only when it escapes the exponential span, and otherwise it \emph{produces} the $\Ei$ term. In practice $\fE_P$ is replaced by the group generated by the exponential sources of the tower (Section~\ref{sec:sources}); the certificate is then relative to that group, and is absolute under the further hypothesis
\begin{enumerate}
\item[(H2)] \emph{Sources.} $\fE/\Fbar^*$ is generated, up to finite index, by the exponential sources of Section~\ref{sec:sources}.
\end{enumerate}
(H2) is the Risch structure theorem \cite{RothsteinCaviness79} in the form of Hebisch's $l$-closed fields \cite[Lemmas~4, 5, 12]{Hebisch18} extended by the Kummer closure of the radical; the implementation reports the sources relative to which a certificate is issued.

\begin{corollary}[Degree bounds and special exponents]\label{cor:bounds}
Let $\cK=\{K_1,\dots,K_k\}$ be a finite set of kernels over $\Fbar L$ and suppose $f=Dg+\sum_ic_iDu_i/u_i+\sum_l\gamma_lK_l$. Then \cite[Thm.~8.7 and Cor.~8.9]{PartII} hold with $v(f)$ replaced by $\min\bigl(v(f),\min_lv(K_l)\bigr)$, under the same hypotheses on the place. In particular Algorithm~6 of Part~II returns a proved bound for the rational part of an antiderivative in $\GQ$ exactly where it returns one for elementary antiderivatives.
\end{corollary}

\begin{proof}
The proof of \cite[Thm.~8.7]{PartII} uses the integrand only through $v(f)$ and the decomposition $f=Dg+\Sigma$ with $v(\Sigma)\ge s_v$. Apply it to $f-\sum_l\gamma_lK_l$.
\end{proof}

Theorem~\ref{thm:structureG} and Corollary~\ref{cor:bounds} are why the extension is small: the denominator theory, the logand theory, the special-prime classification and the degree bounds of Part~II are reused without change, and the new component is the generation of kernel candidates.

\section{Bounds at the places}\label{sec:bounds}

Corollary~\ref{cor:bounds} carries the degree bounds of Part~II over to the kernels, and with them their dichotomy: at a place where the hypotheses of \cite[Thm.~8.7]{PartII} are decided by the criteria (K1)--(K4) of \cite[Prop.~8.11]{PartII}, the bound is proved; elsewhere Algorithm~6 of Part~II takes the classical guess of \cite[\S10.3]{Bronstein05} and raises it on retry. \cite[Rem.~8.14]{PartII} lists what remains heuristic: the kernel degrees at a place where $\rho_v$ or some $\lambda_k$ is not a $\partial_v$-constant; the chains at a lower primitive on which a later generator depends; the second form of the cancellation argument at a place whose residue field carries constants beyond $\Fbar$; and, in the implementation, the special primes over the curve variable that are branch primes. Towers carrying special functions meet these places constantly --- an exponential over a logarithm, a curve carried by an exponential, a non-monomial generator --- and eight of the eighteen examples of Section~\ref{sec:examples} that reach the extended system have a guessed bound. This section decides most of them.

We use the notation of \cite[\S8.2]{PartII}: a discrete valuation $v$ of $\Fbar L$ with uniformiser $\pi$, residue field $\kappa_v$ and a coefficient field $C\subset\Fbar L$ mapping isomorphically onto $\kappa_v$, generated by $c_1,\dots,c_l$; the shift $s=s_v=\min_h\bigl(v(Dh)-v(h)\bigr)$; the leading part $\partial_v\bar c=\overline{\pi^{-s}Dc}$, a derivation of $\kappa_v$; $r=\sigma_v(\pi)-s$ and $\rho_v^{(r)}=\overline{\pi^{-s-1-r}D\pi}$; the kernel degrees $N_v=\{a:\partial_v\bar c=-a\rho_v\bar c\text{ for some }\bar c\ne0\}$; and a generator $c_k$ \emph{attains} the shift if $v(Dc_k)=s$. The four kinds of place (P$_\infty$), (P$_\sigma$), ($\infty_i$), (M) are those of \cite[\S8.2]{PartII}, and \cite[Lemma~8.10]{PartII} reads $s$, $\partial_v$ and $\rho_v^{(r)}$ off the generators. Throughout, $\mu:=\min\bigl(v(f),\min_lv(K_l)\bigr)$ is the valuation of Corollary~\ref{cor:bounds}.

\subsection{Extra constants}

In \cite[Thm.~8.7(ii)]{PartII} the leading form of the rational part is a $\partial_v$-constant, lifted to a constant of $L$; this needs $\Const_{\partial_v}(\kappa_v)=\Fbar$. When a non-attaining generator survives in the residue field its reduction is a new $\partial_v$-constant, and its lift is not a constant of $L$: its derivative enters the cancellation argument, at an order controlled by how far it misses the shift.

\begin{theorem}[Bounds with extra constants]\label{thm:extra}
Let $D$ be $v$-bounded with shift $s$, $r\ge1$, and let the generators of $C$ be split into the attaining $c_k$, $k\in A$, and the non-attaining $c_k$, $k\in B$. Suppose $\Const_{\partial_v}(\kappa_v)=C'$, where $C'$ is the reduction of $C_B:=\Fbar(c_k:k\in B)$. Put
\[
r^*:=\min\Bigl(r,\ \min_{k\in B}\bigl(v(Dc_k)-s\bigr)\Bigr)\ \ge1,\qquad
\rho^*:=\overline{\pi^{-s-1-r^*}D\pi},\qquad
\partial^*\bar c:=\overline{\pi^{-s-r^*}Dc}\quad(c\in C_B),
\]
so that $\rho^*=\rho_v^{(r)}$ if $r^*=r$ and $\rho^*=0$ otherwise, and $\partial^*:C'\to\kappa_v$ is a derivation. Suppose
\begin{enumerate}
\item[(ii*)] for every $a<0$ and every $\bar c\in C'^*$, $\ a\rho^*\bar c+\partial^*\bar c\notin\partial_v(\kappa_v)$.
\end{enumerate}
If $f=Dg+\sum_ic_iDu_i/u_i+\sum_l\gamma_lK_l$, then $v(g)\ge\min\bigl(-r^*,\ \mu-s-r^*\bigr)$. With $B=\emptyset$ this is \cite[Thm.~8.7(ii)]{PartII}: then $C'=\Fbar$, $\partial^*=0$, $r^*=r$, and \textup{(ii*)} reads $\rho_v^{(r)}\notin\partial_v(\kappa_v)$.
\end{theorem}

\begin{proof}
$\partial^*$ is well defined and a derivation: for $c\in C_B$, $Dc=\sum_k(\partial c/\partial c_k)Dc_k$ with $\partial c/\partial c_k\in C_B$ of valuation $\ge0$ and $v(Dc_k)\ge s+r^*$, so $v(Dc)\ge s+r^*$, and the Leibniz rule passes to the coefficient of $\pi^{s+r^*}$.

Write $f':=f-\sum_l\gamma_lK_l=Dg+\Sigma$, $\Sigma:=\sum_ic_iDu_i/u_i$, so $v(f')\ge\mu$ and $v(\Sigma)\ge s$. Let $a:=v(g)<0$. We claim $v(Dg)\le a+s+r^*$. Suppose not. Since $v(Dg)>a+s$, the leading form $\bar c\pi^a$ of $g$ lies in the kernel of the leading part; as $r\ge1$, $\rho_v=0$, so $\partial_v\bar c=0$ and $\bar c\in C'$. Lift $\bar c$ to $c\in C_B$ and write $g=c\pi^a+g_1$ with $v(g_1)\ge a+1$. Then $D(c\pi^a)=\pi^aDc+ac\pi^{a-1}D\pi$ has valuation $\ge a+s+r^*$ and coefficient $\beta:=a\rho^*\bar c+\partial^*\bar c$ at order $a+s+r^*$. As in the proof of \cite[Thm.~8.7(ii)]{PartII}, for $1\le i<r^*$ the leading coefficient $\bar c_i$ of $g_i$ at order $a+i$ must satisfy $\partial_v\bar c_i=0$, so $\bar c_i\in C'$, and $g_{i+1}:=g_i-c_i\pi^{a+i}$ with a lift $c_i\in C_B$; now $D(c_i\pi^{a+i})$ has valuation $\ge a+i+s+r^*>a+s+r^*$, so the coefficient of $Dg_i$ at order $a+s+r^*$ stays $-\beta$. For $g_{r^*}$, $v(g_{r^*})\ge a+r^*$: if $v(g_{r^*})>a+r^*$ that coefficient is $0$, and otherwise it is $\partial_v\bar c_{r^*}$ by \cite[Lemma~8.6]{PartII}. In both cases $\beta\in\partial_v(\kappa_v)$, contradicting \textup{(ii*)}. So $v(Dg)\le a+s+r^*$. If $a<-r^*$ then $v(Dg)<s\le v(\Sigma)$ and $v(f')=v(Dg)$, whence $a\ge v(f')-s-r^*\ge\mu-s-r^*$; if $-r^*\le a<0$ the bound $a\ge-r^*$ holds.
\end{proof}

\subsection{Deciding the hypotheses}

\begin{proposition}[Criteria (K5)--(K8) and branch places]\label{prop:K567}
Let $v$ be a place of one of the four kinds, $s=s_v$, $r=\sigma_v(\pi)-s$, and suppose that $y\notin C$ or that $y$ is non-attaining with $v(Dy)>s+r$.
\begin{enumerate}
\item[(K5)] \textup{(Exponentials over extra constants.)} Suppose $r\ge1$; every attaining generator is a hyperexponential upper generator with $\lambda_k=\mu_k\bar c_k$, $\mu_k\in\Fbar^*$, the $\mu_k$ linearly independent over $\Q$; every non-attaining $c_k$ has $v(Dc_k)\ge s+r$, with coefficient $\beta_k\in\Fbar$ at order $s+r$, and the $\bar c_k$, $k\in B$, are algebraically independent over $\Fbar$; and $\rho_v^{(r)}\in\Fbar^*$. Then $\Const_{\partial_v}(\kappa_v)=C'$, $r^*=r$, and \textup{(ii*)} holds.
\item[(K6)] \textup{(A rational residue field.)} Suppose $r=0$ and exactly one generator $c$ attains the shift, with $\lambda:=\partial_v\bar c\in\Fbar(\bar c)^*$ and $\rho_v\in\Fbar(\bar c)$. Then $N_v\cap\Z_{<0}=\emptyset$ --- and \cite[Thm.~8.7(i)]{PartII} applies --- if and only if $R:=\rho_v/\lambda$ has a nonzero polynomial part, or a multiple pole, or a residue that is not a rational number.
\item[(K7)] \textup{(Two primitive layers.)} Suppose $r=0$ and exactly two generators $c_1,c_2$ attain the shift, both primitive, with $\lambda_1\in\Fbar^*$, $\lambda_2\in\Fbar(\bar c_1)$ having a nonzero residue, and $\rho_v\in\Fbar^*$. Then $N_v=\{0\}$.
\item[(K8)] \textup{(Exponentials over primitives.)} Suppose $r=0$, the attaining generators are primitives $c_k$ with $\partial_v\bar c_k\in\Fbar^*$ and hyperexponentials $\theta_j$ with $\partial_v\bar\theta_j=\lambda_j\bar\theta_j$, $\lambda_j$ a polynomial in the $\bar c_k$, $\rho_v\in\Fbar^*$, and no rational combination of the $\lambda_j$ is a nonzero constant. Then $N_v=\{0\}$.
\item[(B)] \textup{(Branch places over special primes.)} Let $m=2$ and let $v$ be the ramified place over a special prime $\sigma$ with $\sigma\mid q$, $q=\sigma q_1$. Then $v(a+by)=\min(2v_\sigma(a),2v_\sigma(b)+1)$ exactly, $y$ is a uniformiser, and the class of a scalar $e$ modulo the place at order $2k$ is the class of $e/(\sigma q_1)^k$; with these data \cite[Prop.~8.11]{PartII} and \textup{(K5)--(K7)} apply as stated.
\end{enumerate}
\end{proposition}

\begin{proof}
(K5) By \cite[Lemma~8.10]{PartII}, $\kappa_v=C'(\bar c_k:k\in A)$ with $\partial_v\bar c_k=\mu_k\bar c_k$ and $\partial_v=0$ on $C'$ (the $c_k$, $k\in B$, miss the shift). This is a tower of hyperexponential monomials over the field $C'$ with $\Q$-independent $\mu_k\in\Fbar$, so no new constants arise \cite[Thm.~5.1.2]{Bronstein05}: $\Const_{\partial_v}(\kappa_v)=C'$. Every non-attaining generator has $v(Dc_k)\ge s+r$, so $r^*=r$ and $\rho^*=\rho_v^{(r)}\in\Fbar^*$. The derivation $\partial^*=\sum_{k\in B}\beta_k\,\partial/\partial\bar c_k$ has constant coefficients and maps $C'$ into itself, so $\beta=a\rho^*\bar c+\partial^*\bar c\in C'$. An element of $\partial_v(\kappa_v)$ has zero component of degree $0$ in the Laurent expansion in the $\bar c_k$, $k\in A$ (the argument of (K4) in \cite[Prop.~8.11]{PartII}), so $\partial_v(\kappa_v)\cap C'=\{0\}$ and \textup{(ii*)} fails only if $\partial^*\bar c=-a\rho^*\bar c$ for some $\bar c\in C'^*$. But a nonzero constant is not a logarithmic derivative for a derivation with constant coefficients on $\Fbar(\bar c_k:k\in B)$: writing $\bar c$ as a quotient of products of irreducible polynomials $p$, one needs $p\mid\partial^*p$ for each, and $\partial^*$ lowers the total degree, so $\partial^*p=0$ and $\partial^*\bar c/\bar c=0$.

(K6) The other generators do not attain the shift, so they reduce to $\partial_v$-constants, $C'$, and $\kappa_v=C'(\bar c)$ with $\partial_v=\lambda\,d/d\bar c$. Then $a\in N_v$ if and only if $-aR$ is a logarithmic derivative for $d/d\bar c$ in $C'(\bar c)$. By partial fractions, the logarithmic derivatives are exactly the rational functions with zero polynomial part, only simple poles and integer residues. Since $R\in\Fbar(\bar c)$, its polynomial part, poles and residues are the same over $C'$; and $-aR$ has integer residues for some $a<0$ if and only if the residues of $R$ are rational (take $-a$ a common denominator).

(K7) As in (K6), $\kappa_v=C'(\bar c_1,\bar c_2)$ with $\partial_v\bar c_1=\lambda_1$, $\partial_v\bar c_2=\lambda_2$. Since $\lambda_2$ has a nonzero residue it is not a derivative in $C'(\bar c_1)$, so $\bar c_2$ is a primitive monomial over $C'(\bar c_1)$. If $\partial_v\bar e/\bar e=-a\rho_v\in\Fbar^*$ with $\bar e\in\kappa_v^*$, then $\bar e\in C'(\bar c_1)$, because every irreducible polynomial in a primitive monomial is normal and a normal factor would put a pole into the logarithmic derivative \cite[Thm.~5.1.1]{Bronstein05}; and then $d\bar e/d\bar c_1=-(a\rho_v/\lambda_1)\,\bar e$ with a nonzero constant coefficient, which is impossible by (K6).

(K8) If $\partial_v\bar e/\bar e=-a\rho_v$ with $\bar e\in\kappa_v^*$, write $\bar e=\bar r\prod\bar\theta_j^{n_j}$ with $\bar r$ in the field of the primitives. This is possible because the $\bar\theta_j$ are hyperexponential over it and a logarithmic derivative in $\kappa_v$ that lies in that field has this form, by the Risch structure theorem \cite{RothsteinCaviness79}. Then $-a\rho_v=\partial_v\bar r/\bar r+\sum n_j\lambda_j$. The primitives have constant leading coefficients, so, as in (K5), $\partial_v\bar r/\bar r$ has no constant term unless it vanishes. Hence $\sum n_j\lambda_j$ is the nonzero constant $-a\rho_v$, which is excluded.

(B) The two coordinates of $a+by$ have valuations of different parity at the ramified place, so no cancellation occurs; $v(y)=1$; and $y^{2k}=\sigma^kq_1^k$ with $q_1$ a unit there. \cite[Lemma~8.10]{PartII} and the proofs above use only these data.
\end{proof}

\begin{example}[The heuristic places of Section~\ref{sec:examples}]\label{ex:boundsdecided}
(i) $\int e^x\log x\,dx$, $L=\Q(x,t,\theta)$, $t=\log x$, $\theta=e^x$, at $x=\infty$: $\pi=1/x$, $D\pi=-\pi^2$, so $\sigma_v(\pi)=1$; $\theta$ attains $s=0$ with $\mu_\theta=1$; $t$ misses it, $Dt=\pi$, so $B=\{t\}$, $r=r^*=1$, $\beta_t=1$ and $\rho_v^{(1)}=-1$. (K5) applies: $C'=\Fbar(\bar t)$, $\partial^*=d/d\bar t$, and (ii*) asks that $\bar c'=a\bar c$ have no rational solution $\bar c\ne0$. The degree in $x$ is proved $\le1$ where the guess was $3$. The same place in $\int\log(x)e^{-x^2}\,dx$ has $s=-1$, $r=2$, and (K5) proves the bound $2$; every bound in force for that integrand is now proved (\S\ref{ex:certificate}). At $\theta_1=\infty$ in $\int e^{e^x}\,dx$ the attaining generator is $\theta_2$, with $\mu=1$, and the non-attaining one is $x$, with $\beta_x=1$.

(ii) $\int\sqrt{\log x}\,dx$, $L=\Q(x,z)$, $Dz=1/(2xz)$, at $x=\infty$: $s=1=\sigma_v(\pi)$, $r=0$, $z$ attains with $\lambda=1/(2\bar z)$, $\rho_v=-1$, so $R=-2\bar z$ has a polynomial part and (K6) gives $N_v=\{0\}$; at the special prime $(x)$, $s=-1$, $\lambda=1/(2\bar z)$, $\rho_v=1$ and $R=2\bar z$. $z$ is not a monomial, and (K1) does not apply.

(iii) The Lambert tower at $x=\infty$ and at $(x)$: $R=-1-1/\bar W$ and $R=1+1/\bar W$ have polynomial parts, and (K6) decides both. At $(W)$ itself $R=1/\bar x$ has the rational residue $1$: every $a\le0$ is a kernel degree (Remark~\ref{rem:kerexp}).

(iv) $\int\log\log x\,dx$ at $x=\infty$ and at $(x)$: $t_1=\log x$ and $t_2=\log t_1$ both attain, $\lambda_1=1$, $\lambda_2=1/\bar t_1$, $\rho_v=-1$ and $\rho_v=1$ respectively, and (K7) applies.

(v) The branch special prime $(\theta)$ of $y^2=\theta^3-\theta$, $\theta=e^x$ (\S\S\ref{ex:sinh}, \ref{ex:sinh2}): $y$ is the uniformiser, $Dy/y=\tfrac12Dq/q$ has class $\tfrac12$, $x$ attains with $\lambda=1$, and (K1) applies through (B), proving the exponent $1$ of $\theta$ in \S\ref{ex:sinh2}. For $y^2=2x\theta$ (\S\ref{ex:hebisch22}) the class is $\rho_v=1/(2\bar x)-1/2$, and (K6) applies through (B).
\end{example}

A special prime split over $\Fbar$ is decided at each factor rather than through its parent: the factors are conjugate, and so is the data of Lemma~8.10 of Part~II at their places. This decides the hypertangent special $1+t^2$ of the Fresnel tower, where (K2) applies at $t\mp\sqrt{-1}$ and not at the irreducible $1+t^2$.

\subsection{What remains heuristic, and what it means here}

\begin{proposition}[Kernel degrees of exponential elements]\label{prop:kerexp}
Let $\xi\in\fE$ with $n:=v(\xi)\ne0$ and $v\bigl(D\lambda(\xi)\bigr)>s_v$. Then $n\Z\subseteq N_v$; in particular $N_v\cap\Z_{<0}\ne\emptyset$ and hypothesis \textup{(i)} of \cite[Thm.~8.7]{PartII} fails at $v$.
\end{proposition}

\begin{proof}
$D(\xi^k)=k\,\xi^k\,D\lambda(\xi)$ has valuation $\ge kn+s_v+1$, so the leading form of $\xi^k$, of degree $kn$, lies in the kernel of the leading part.
\end{proof}

\begin{remark}\label{rem:kerexp}
The condition $v(D\lambda(\xi))>s_v$ says that $\lambda(\xi)$ is constant to first order at $v$ --- it is the silent-zero condition of Section~\ref{sec:candidates} read at the level of the leading part. At the special prime $(W)$ of the Lambert tower, $\xi=e^{W}=x/W$ has $v(\xi)=-1$ and $\lambda=W$, and the kernel degrees are its powers: the class $\bar x^{-a}W^a$ is the leading form of $\xi^{-a}$. So the places that Proposition~\ref{prop:K567} leaves undecided are, in the towers of this paper, exactly the places at which the silent $\Ei$ kernels live, and deciding the exponent there is a Risch equation in $(\kappa_v,\partial_v)$. In keeping with the series we leave it to the guess; the silent kernel column is what the guess protects.
\end{remark}

The other places still left to the guess in Section~\ref{sec:examples} are: the chain at $t_1=\infty$ in $\int\log\log x\,dx$ ($t_2$ depends on $t_1$; \cite[Ex.~8.13]{PartII}); and $x=\infty$ in \S\ref{ex:hebisch22}, a place of a radicand in two generators, which is none of the four kinds of place.

\subsection{Effect}

With (K5)--(K8) and (B) installed in Algorithm~6 for the special stage, the eighteen examples of Section~\ref{sec:examples} that reach the extended system have every bound proved in $14$ cases instead of $8$; their ans\"atze shrink from $532$ unknowns to $292$ and the rungs of the retry ladder from $23$ to $19$. On $141$ random exact integrands $f=Dg$ in nine towers --- logarithm and exponential, nested exponentials, Lambert, $\sqrt{\log x}$, nested logarithms, $\log x$ with $e^{-x^2}$, and the curves $y^2=\theta^3-\theta$, $y^2=2x\theta$, $y^2=\theta^2+\theta+1$ over exponentials --- Part~II with the extended criteria integrates every one, as it must when every proved bound is correct, with $4813$ unknowns against $7991$. The script \texttt{bounds\_special.py} reproduces these numbers.

\section{Kernel candidates}\label{sec:candidates}

\subsection{Exponential sources}\label{sec:sources}

An \emph{exponential source} is a pair $(\xi,\lambda)$ of elements of $L$ with $D\xi/\xi=D\lambda$. The tower provides them: a hyperexponential generator $\theta=e^{\eta}$ gives $(\theta,\eta)$; a logarithmic generator $t=\log a$ gives $(a,t)$; a Lambert generator $W=W(a)$ gives $(a/W,W)$, since $We^W=a$; a hypertangent generator $t=\tan a$ gives $\bigl((1+\sqrt{-1}\,t)/(1-\sqrt{-1}\,t),\ 2\sqrt{-1}\,a\bigr)$. When a generator was absorbed by flattening \cite[Lemma~3.2]{PartII} its expression in the new generators is substituted: for $\sqrt{\log x}$ the tower is $\Q(x,z)$ with $z^2=\log x$, and the logarithm gives the source $(x,z^2)$. Every source is verified by the identity $D\xi/\xi=D\lambda$ in the tower, and normalised by $\xi=e^{\lambda}$, which is exact for the four kinds above. Products of sources, $\xi=\prod\xi_i^{j_i}$ with $\lambda=\sum j_i\lambda_i$, are the \emph{exponential monomials}.

\subsection{$\Ei$ kernels determined by residues}\label{sec:eires}

Let $P$ be a normal prime of $L$ over a prime linear in its main generator $g$, $\ell=g-\rho$, and let $f$ have a pole of order at least $\delta_P$ at $P$. The canonical residue $\tau=\hat\tau_P(f)$ is computed by \cite[Alg.~2]{PartII} (at a critical pole it is $(f\ell/D\ell)|_{g=\rho}$). For every exponential monomial $\xi$ whose logarithm $\lambda$ takes a constant value $c$ at $g=\rho$, let
\[
K_\xi:=\xi\,\frac{Dv}{v},\qquad v=\lambda-c,
\]
whose residue at $P$ is $\tau_\xi=v_P(v)\,\bar\xi$ by Proposition~\ref{prop:kernelnormal}(d). One solves the linear system
\[
\tau\;=\;r_0+\sum_\xi r_\xi\,\tau_\xi
\]
for constants $r_0,r_\xi$ by equating coefficients in $\kappa(P)$. If it has no solution, Corollary~\ref{cor:rescert} applies. Otherwise the kernels $r_\xi K_\xi$ are subtracted from $f$: the remainder has the constant residue $r_0$ at $P$, which Part~II realises by logarithms. By Lemma~\ref{lem:analytic} and the normalisation $\xi=e^{\lambda}$, each subtracted kernel integrates to $e^{c}\,\Ei(\lambda-c)$. The argument $v$ is determined by $\xi$ and $P$; nothing needs to be found, and in contrast with the logarithmic part there is no torsion question. A prime of $F[g]$ with constant coefficients and degree $\ge2$ is split over $\Fbar$, and each root is treated as above (\S\ref{ex:complexei}).

\subsection{$\Ei$ kernels with silent zeros}

An $\Ei$ argument all of whose zeros lie at special primes or at places at infinity leaves no residue at a normal prime. For each exponential monomial $\xi$ and each special prime or place at infinity at which $\lambda(\xi)$ takes a constant value $c$, the kernel $K_\xi$ with $v=\lambda-c$ is added as a column with an unknown coefficient: $\Ei(e^x)$ at the special prime $(e^x)$ (\S\ref{ex:eiexp}), $\Ei(W)$ at the special prime $(W)$ of the Lambert tower (\S\ref{ex:lambert}), $\Ei(1/x)$ at $x=\infty$, $\Ei(x)$ at the special prime $(x)$ of $\log x$ (\S\ref{ex:explog}). The special primes are those of \cite[Thm.~6.1]{PartII}.

\subsection{$\Gamma(s,\cdot)$ kernels, $s>0$: the Kummer search}\label{sec:kummer}

These kernels leave no residue (Proposition~\ref{prop:kernelnormal}(e)), and their candidates come from the exponential data and Corollary~\ref{cor:parameter}. Let $\theta$ be a hyperexponential generator and $J$ an exponent of $\theta$ in $f$ --- a half-integer when $q=c\,\theta^a q_0$ with $a$ odd, $y$ then carrying the weight $a/2$ --- or, for the other sources, $J=\pm1$. Let $\Lambda=J\lambda$. When $\Lambda$ is a rational function of a single generator $g$, the argument has the form $v=\Lambda+c$ with a constant shift $c$ taken from a finite set:
\begin{enumerate}
\item the roots $c$ of $\operatorname{Res}_g\bigl(N+cD_\Lambda,\ \partial_g(N+cD_\Lambda)\bigr)$, $\Lambda=N/D_\Lambda$, for which $\Lambda+c$ has a multiple zero;
\item the values $c=-\Lambda(\rho)$ at the roots $\rho$ of $q$ in $g$, for which $\Lambda+c$ vanishes at a branch prime.
\end{enumerate}
For each $c$, let $N_c$ be the gcd of the multiplicities of the zeros of $v$, a zero at a branch prime counted $m$ times. The admissible parameters are $s=a/k$ with $k\mid N_c$ (Corollary~\ref{cor:parameter}). For each $s$ the \emph{Kummer test} decides whether some $\omega\in L$ has $\omega^k=\kappa\cdot W$ with $W=(-v)^{k(s-1)}\theta^{kJ}$ and $\kappa$ constant: either $W$ is a constant times a $k$-th power in $K_0$, and $\omega$ is its root, or ($m=k=2$) $W/q$ is a constant times a square, and $\omega=y\sqrt{W/q}$. Every candidate is then verified by \eqref{eq:kernel} as an identity in $L$ and enters as a column $K_s(v,\omega)=\omega\,Dv$.

\begin{lemma}[Kummer test]\label{lem:kummer}
If $\omega\in L$ passes the test, $(v,\omega)$ is an $s$-kernel datum and $\int K_s(v,\omega)=\kappa^{-1/k}e^{-c}\,\Gamma(s,-v)$ for a choice of the $k$-th root.
\end{lemma}

\begin{proof}
$\omega^k=\kappa(-v)^{k(s-1)}\theta^{kJ}$ gives $k\,D\omega/\omega=k(s-1)Dv/v+kJ\,D\theta/\theta=k(s-1)Dv/v+k\,D\Lambda$, and $D\Lambda=Dv$: this is \eqref{eq:kernel}. With $\theta^J=e^{\Lambda}=e^{v-c}$, $\omega=\kappa^{1/k}(-v)^{s-1}e^{v}e^{-c}$, and Lemma~\ref{lem:analytic} applies.
\end{proof}

Every step of the search is algebraic: no differential equation is solved, and a candidate is admitted only after \eqref{eq:kernel} is verified. The search is where the radical matters: by Corollary~\ref{cor:parameter}, ramification is what lets $\Gamma(\tfrac12,\cdot)$ appear with an argument that has only simple zeros in $K_0$ (\S\S\ref{ex:conic}, \ref{ex:hebisch22}).

\subsection{Hypertangent towers}

For a hypertangent generator $t$, the source $\xi_t=(1+\sqrt{-1}\,t)/(1-\sqrt{-1}\,t)$ is supported on the special primes $(t\mp\sqrt{-1})$, and the rules above produce $\Si$, $\Ci$ and the Fresnel integrals as complex $\Ei$ and $\erf$ terms, recombined into real form at output. The residue decomposition has a transparent shape: $\hat\tau_P(f)\in\Fbar(\bar t)$ lies in the exponential span exactly when it is a Laurent polynomial in $\bar\xi_t$, i.e.\ when its poles in $\bar t$ lie at $\bar t=\pm\sqrt{-1}$ (\S\ref{ex:si}).

\section{Elliptic kernels}\label{sec:elliptic}

\subsection{The points live on the radical}

\begin{proposition}[Points]\label{prop:points}
Let $\cE$ be a constant elliptic curve. Then $\cE(\Fbar L)=\cE(\Fbar K_0(y))$, and if the tower has no radical layer, $\cE(\Fbar L)=\cE(\Fbar)$. Consequently no elliptic integral occurs in the antiderivative of an element of a transcendental tower, Liouvillian or not, monomial or not.
\end{proposition}

\begin{proof}
If $t$ is transcendental over a field $M$, a point of $\cE$ over $M(t)$ is a rational map from $\mathbb P^1_M$ to a curve of genus $1$, hence constant: $\cE(M(t))=\cE(M)$. The upper generators are algebraically independent over $K_0(y)$ by (T3), and $K_0=F(t_1,\dots,t_j)$ is purely transcendental over $F$; iterate. A constant point makes every elliptic term of Theorem~\ref{thm:ellliouville} vanish.
\end{proof}

Every elliptic integral in an answer is therefore carried by the radical. The basic source is a \emph{pencil}: $q=Q(u)$ with $u$ a generator of $K_0$ and $Q$ with constant coefficients, so that $(u,y)$ is a point of the constant curve $\cC:Y^2=Q(U)$. The pencil variable may be $x$, an exponential, a logarithm, or any generator (\S\ref{ex:sinh}).

\subsection{Columns}

\begin{proposition}[Finitely many columns]\label{prop:columns}
Let $u$ be a pencil with curve $\cC$ of genus $g$ and let $\omega_1,\dots,\omega_{2g}$ be differentials of the second kind on $\cC$ forming a basis of $H^1_{\mathrm{dR}}(\cC)$. For every constant elliptic curve $\cE$, every morphism $\phi:\cC\to\cE$ and every differential $\omega$ of the first or second kind on $\cE$, with $P_0=(u,y)$,
\[
(\phi\circ P_0)^*\omega\ \in\ \sum_{i=1}^{2g}\Fbar\,P_0^*\omega_i\;+\;D(\Fbar L).
\]
\end{proposition}

\begin{proof}
$\phi^*\omega$ is of the second kind on $\cC$, hence congruent modulo $d\Fbar(\cC)$ to a combination of the $\omega_i$; and $P_0^*(dh)=D(h\circ P_0)$ with $h\circ P_0\in\Fbar L$.
\end{proof}

For $y^2=Q(u)$ of degree $d$ the differentials $U^i\,dU/Y$, $0\le i\le d-2$, span the second-kind part together with, for even $d$, the third-kind differential at the two points at infinity; the algorithm uses the $d-1$ columns $u^iDu/y$. On a cubic these are the two columns $Du/y$ and $u\,Du/y$.

\subsection{Residues pull back}

\begin{lemma}[Residues pull back]\label{lem:respull}
Let $\cC$ be a constant curve, $P\in\cC(\Fbar L)$ a non-constant point, $\omega$ a differential on $\cC$ over $\Fbar$, $\fp$ a normal prime and $\bar P\in\cC(\kappa(\fp))$ the reduction of $P$. If $\bar P$ is not an $\Fbar$-point, $v_\fp(P^*\omega)\ge1-\delta_\fp$ and $\hat\tau_\fp(P^*\omega)=0$. If $\bar P=Q\in\cC(\Fbar)$, let $z$ be a local parameter at $Q$ over $\Fbar$, $n:=v_\fp(z\circ P)\ge1$, and $k$ the pole order of $\omega$ at $Q$. Then $\hat\tau_\fp(P^*\omega)=n\,\res_Q\omega\in\Fbar$, $v_\fp(P^*\omega)\ge-(k-1)n-\delta_\fp$ for $k\ge1$, and $v_\fp(P^*\omega)\ge1-\delta_\fp$ if $\omega$ is regular at $Q$.
\end{lemma}

\begin{proof}
The poles of $\omega$ are $\Fbar$-points. In the first case $\omega=h\,dz'$ near $\bar P$ with $h$ regular and $z'$ a local parameter, so $P^*\omega=(h\circ P)D(z'\circ P)$ with both factors $\fp$-integral and $v_\fp(D(z'\circ P))\ge1-\delta_\fp$. In the second, choose $h\in\Fbar(\cC)$ with $\omega-dh$ having at most a simple pole at $Q$; then $\omega-dh=\res_Q(\omega)\,dz/z+h_1\,dz$ with $h_1$ regular at $Q$, and
\[
P^*\omega=D(h\circ P)+\res_Q(\omega)\,\frac{D(z\circ P)}{z\circ P}+(h_1\circ P)\,D(z\circ P).
\]
The canonical residue kills the first and third terms and gives $n\res_Q\omega$ on the second. The valuation bound is the same computation with $\omega=z^{-k}(\text{unit})\,dz$.
\end{proof}

\begin{corollary}\label{cor:ellnormal}
At a normal prime the first-kind kernels have valuation $\ge1-\delta_\fp$ and no residue --- they share the window with the $\Gamma(s,\cdot)$ kernels, $s>0$ --- the second-kind kernels have no residue and poles deeper than $\delta_\fp$ only where $\bar P$ is a pole of the differential, and the third-kind kernels have constant residues. Hence the residue certificate of Corollary~\ref{cor:rescert} remains a certificate when elliptic integrals are allowed, and the bounds of Theorem~\ref{thm:structureG}(i) and Corollary~\ref{cor:bounds} hold with the minimum taken over the elliptic kernels as well.
\end{corollary}

\subsection{The third kind}

In Part~II a residue divisor proved non-torsion by reduction modulo primes \cite[Prop.~9.4]{PartII} certifies non-elementarity. In the elliptic class it is instead a third-kind term, and on a pencil the term is explicit.

\begin{proposition}[Third-kind kernel]\label{prop:third}
Let $u$ be a pencil with $Du=r\in F(u)^*$ and $q=Q(u)$ squarefree, let $f=f_0+f_1y$ with $f_1\in F(u)$, and let $p\in F[u]$ be an irreducible factor of the denominator of $f_1$ of multiplicity one with $p\nmid q$. Let $c\in F[u]$, $\deg c<\deg p$, be the reduction of $f_1pq/r$ modulo $p$. Then
\[
K_p:=\frac{c\,Du}{p\,y}=\frac{c\,r}{p\,q}\,y
\]
has, at every place over $p$, the same residue as $f$ minus the residue of $f_0$; so $f-K_p$ has residues at the places over $p$ that are equal on the two sheets, realised by $\log p$. With $\rho$ ranging over the roots of $p$, $\int K_p=\sum_\rho\frac{c(\rho)}{p'(\rho)}\int\frac{dU}{(U-\rho)Y}$.
\end{proposition}

\begin{proof}
At a place $(\rho,\pm y_\rho)$ over $p$ the residue of $a\,y$ with $a$ having a simple pole is $\pm(a\,p/Dp)|_\rho\,y_\rho=\pm y_\rho(a\,p/(p'r))|_\rho$ by \cite[Thm.~7.4]{PartII}. For $a=f_1$ and $a=cr/(pq)$ these agree because $c\equiv f_1pq/r$ modulo $p$ and $y_\rho^2=q(\rho)$. The residues of $f_0$ at the two places over $\rho$ are equal. The last formula is the partial-fraction decomposition of $c/p$.
\end{proof}

\subsection{Presentation}\label{sec:presentation}

\begin{lemma}[Legendre forms]\label{lem:legendre}
Let $y^2=(U-e_1)(U-e_2)(U-e_3)$, $\lambda^2=e_1-e_3$, $m=(e_2-e_3)/(e_1-e_3)$, and $\sin\varphi=\lambda/\sqrt{U-e_3}$. Then, up to additive constants,
\[
\int\frac{dU}{y}=-\frac2\lambda F(\varphi\,|\,m),\qquad
\int\frac{U\,dU}{y}=-\frac2\lambda(e_3+\lambda^2)F(\varphi\,|\,m)+2\lambda E(\varphi\,|\,m)+\frac{2y}{U-e_3},
\]
\[
\int\frac{dU}{(U-\rho)\,y}=\frac{2}{\lambda(\rho-e_3)}\Bigl(F(\varphi\,|\,m)-\Pi(n;\varphi\,|\,m)\Bigr),\qquad n=\frac{\rho-e_3}{\lambda^2}.
\]
For $y^2=c(1-U^2)(1-kU^2)$, $\int dU/y=F(\arcsin U\,|\,k)/\sqrt c$ and $\int U^2\,dU/y=\bigl(F-E\bigr)(\arcsin U\,|\,k)/(k\sqrt c)$.
\end{lemma}

\begin{proof}
With $U=e_3+\lambda^2/\sin^2\varphi$ one has $dU/y=-(2/\lambda)\,d\varphi/\Delta$, $\Delta=\sqrt{1-m\sin^2\varphi}$, which gives the first formula. For the second, $\int d\varphi/(\sin^2\varphi\,\Delta)=F-E-\cot\varphi\,\Delta$ (differentiate), and $\cot\varphi\,\Delta=y/(\lambda(U-e_3))$. For the third, $1/(U-\rho)=\sin^2\varphi/(\lambda^2(1-n\sin^2\varphi))$ and $\sin^2\varphi/(1-n\sin^2\varphi)=n^{-1}\bigl((1-n\sin^2\varphi)^{-1}-1\bigr)$. The Legendre case is the definition of $F$ and $E$ with $U=\sin\varphi$.
\end{proof}

The Weierstrass forms ($\int dX/Y=-w$, $\int X\,dX/Y=\zeta(w)$ with $X=\wp(w)$, $Y=-\wp'(w)$) and Jacobi's form of the third kind, $u\,Z(a)+\tfrac12\log\bigl(\Theta(u-a)/\Theta(u+a)\bigr)$ with $\Theta(u)=\vartheta_4(\pi u/(2K))$ and $Z(a)=E(\operatorname{am}a)-\tfrac{E(m)}{K(m)}F(\operatorname{am}a)$ \cite[Ch.~22]{WhittakerWatson}, \cite[Chs.~19, 22]{DLMF}, are the other presentations; the latter is where $\vartheta$ and the complete integrals $K(m)$, $E(m)$ appear, as constants. When the residue component is torsion of order $N$, $a$ is an $N$-torsion point, and the $N$-th power of the theta quotient is, up to the exponential of a first-kind multiple, the Miller function with which Part~II realises torsion logands \cite[Alg.~3(d)]{PartII}: logand and third kind are the two sides of the torsion dichotomy.

For $\Gamma(s,\cdot)$ the presentation writes $\Gamma(\tfrac12,z)=\sqrt\pi\erfc(\sqrt z)$, $\erfc(\sqrt{-1}\,z)$ through $\erfi(z)$, $\Ei(\log a)=\li(a)$, and a pair $a\,\Ei(\sqrt{-1}z)+b\,\Ei(-\sqrt{-1}z)$ as $(a+b)\Ci(z)+\sqrt{-1}(a-b)\Si(z)$, each up to an additive constant.

\section{The algorithm}\label{sec:algorithm}

The algorithm is written as ten boxes in the style of Part~II: nine here, and \textsc{PartialIntegrate} in \S\ref{sec:partial}. Its main loop, Algorithm~S7, is Part~II's Algorithm~4 with kernel columns. Algorithms~S1--S4 generate the kernels, S5 extends Part~II's Algorithm~6, S6 solves and minimises, S8 certifies necessity and S9 presents the answer. Part~II's Algorithms~1--6 are used unchanged and cited by number. Throughout, elements of $L$ are held in the Trager coordinates $\sum_ia_iw_i$ of \cite[Prop.~3.5]{PartII}, and ``the scalar part'' of an element means its $w_0$-coordinate when all other coordinates vanish.

\algheader{Algorithm S1}{$\textsc{Sources}(L)$\quad(* exponential sources; Section~\ref{sec:sources} *)}
\begin{algorithmic}[1]
\Statex (* Given the tower and, when available, its surface form (the pairs (generator, expression) recorded by the tower builder), return a list of exponential sources $(\xi,\lambda)$, $\xi,\lambda\in L$, with $D\xi/\xi=D\lambda$ and the normalisation $\xi=e^{\lambda}$. *)
\State $\mathcal E\gets\emptyset$;\quad $\phi\gets$ the substitution sending each flattened generator $t$ to its expression $(u^{m'}-r)/c$ in the current generators \Comment{Lemma~3.2 of Part~II}
\For{each recorded pair $(t,\,E)$ of the surface form}
\State \textbf{case} $E=\exp(a)$: $(\xi,\lambda)\gets(t,\ \phi(a))$
\State \textbf{case} $E=\log(a)$: $(\xi,\lambda)\gets(\phi(a),\ \phi(t))$
\State \textbf{case} $E=W(a)$ (Lambert): $(\xi,\lambda)\gets(\phi(a)/t,\ t)$ \Comment{$We^W=a$}
\State \textbf{case} $E=\tan(a)$: $(\xi,\lambda)\gets\bigl((1+\sqrt{-1}\,t)/(1-\sqrt{-1}\,t),\ 2\sqrt{-1}\,\phi(a)\bigr)$;\quad \textbf{case} $E=\tanh(a)$: $\bigl((1+t)/(1-t),\ 2\phi(a)\bigr)$
\State \textbf{if} $D\xi/\xi-D\lambda=0$ exactly in $L$: append $(\xi,\lambda)$ to $\mathcal E$
\EndFor
\If{there is no surface form}
\For{each generator $t_k$, $k\ge2$, with $Dt_k=(w_k,0,\dots,0)$}
\State \textbf{if} $Dt_k/t_k$ is free of $t_k$: $\eta\gets$ Part~II's Algorithm~4 on $Dt_k/t_k$ over the tower below $t_k$; \textbf{if} $\eta$ is rational (no logarithm): append $(t_k,\eta)$
\State \textbf{else if} $Dt_k$ is free of $t_k$: $I\gets$ Part~II's Algorithm~4 on $Dt_k$ over the tower below; \textbf{if} $I=\sum n_i\log a_i$ with $n_i\in\Z$: append $(\prod a_i^{n_i},t_k)$
\EndFor
\State keep only the pairs verified by $D\xi/\xi=D\lambda$
\EndIf
\State \textbf{return} $\mathcal E$
\end{algorithmic}
\algfooter

\algheader{Algorithm S2}{$\textsc{ResidueKernels}(f,\mathcal E)$\quad(* $\Ei$ kernels determined by residues; Section~\ref{sec:eires} *)}
\begin{algorithmic}[1]
\Statex (* Return $(f',\mathcal D)$ with $f'=f-\sum_{(r,K)\in\mathcal D}rK$ and $\mathcal D$ the determined $s=0$ kernels, or \textbf{NotInClass}$(P,\tau)$. *)
\State $\mathcal D\gets\emptyset$;\quad $d\gets\operatorname{lcm}_i\den(a_i)$ for $f=\sum_ia_iw_i$
\For{each irreducible factor $p$ of $d$ over $F$ that involves a generator}
\State $(e_P,\eta_P,\delta_P,\mathrm{special}_P)\gets\textsc{Classify}(P)$ (Part~II, Alg.~1); \textbf{skip} $p$ if it is special or a branch prime, or if $v_P(f)>-\delta_P$
\State $g\gets$ the last generator in $p$;\quad $\mathcal R\gets$ the roots of $p$ in $g$: $\{\rho\}$ if $p$ is linear in $g$, the roots over $\Fbar$ if $p\in F[g]$ with constant coefficients, else \textbf{skip} $p$
\For{each $\rho\in\mathcal R$, $\ell\gets g-\rho$}
\State $\tau\gets\bigl(f\ell/D\ell\bigr)\big|_{g=\rho}$ if $v_P(f)=-\delta_P$; else (only when $q$ is absent and $p=\ell$) $\tau\gets\textsc{CanonicalResidue}(P,f)$ (Part~II, Alg.~2) \Comment{Thm.~7.4, Prop.~7.6 of Part~II}
\State \textbf{if} $\tau=(\tau_0,\tau_1,\dots)$ is constant at the places over $P$ ($\tau_0\in\Fbar$ and, for $m=2$, $\tau_1\in\Fbar$ and $q(\rho)\in\Fbar$): \textbf{skip}
\State $e_{\max}\gets$ the maximum of $2$ and the degrees, in every generator, of the numerators and denominators of $\tau$
\State $\mathcal M\gets$ the monomials $\xi=\xi_a^{j}$ and $\xi_a^{j}\xi_b^{k}$ over sources $(\xi_a,\lambda_a),(\xi_b,\lambda_b)\in\mathcal E$ with $1\le|j|,|k|\le e_{\max}$, and $\lambda(\xi)=j\lambda_a$, respectively $j\lambda_a+k\lambda_b$
\State \textbf{for} $\xi\in\mathcal M$ with $\lambda(\xi)$ scalar and $c:=\lambda(\xi)|_{g=\rho}$ finite and in $\Fbar$: $v\gets\lambda(\xi)-c$;\ $K_\xi\gets\xi\,Dv/v$;\ $\tau_\xi\gets(K_\xi\ell/D\ell)|_{g=\rho}$; keep $\xi$ if $\tau_\xi\ne0$
\State solve $\tau=r_0+\sum_\xi r_\xi\tau_\xi$ for $r_0,r_\xi\in\Fbar$, by equating in each coordinate the coefficients of the monomials in the remaining generators of the numerator of $\tau-r_0-\sum r_\xi\tau_\xi$; free unknowns $0$
\State \textbf{if} there is no solution: \textbf{return NotInClass}$(P,\tau)$ \Comment{Cor.~\ref{cor:rescert}, relative to $\mathcal E$}
\State \textbf{for} $r_\xi\ne0$: $f\gets f-r_\xi K_\xi$;\ append $(r_\xi,K_\xi)$ to $\mathcal D$, with antiderivative $r_\xi e^{c}\Ei(v)$ \Comment{Lemma~\ref{lem:analytic}}
\EndFor
\EndFor
\State \textbf{return} $(f,\mathcal D)$
\end{algorithmic}
\algfooter

\algheader{Algorithm S3}{$\textsc{KernelCandidates}(f,\mathcal E)$\quad(* the Kummer search and silent zeros; Sections~\ref{sec:kummer}, 6.3 *)}
\begin{algorithmic}[1]
\Statex (* Return a list $\cK$ of verified kernels $K_s(v,\omega)=\omega\,Dv$, each with $(s,v,\omega,\kappa,c)$ such that $\int K_s=\kappa^{-1/k}e^{-c}\Gamma(s,-v)$, $k$ the denominator of $s$ (for $s=0$: $e^{-c}\Ei(v)$). *)
\State $\Sigma\gets$ the pairs $(g,\rho)$ of the special primes of the tower (the special irreducible factors of the generators and of the numerators and denominators of the $Dt_k$, Part~II Alg.~1) that are linear in a generator $g$, with root $\rho$
\For{each source $(\xi_a,\lambda_a)\in\mathcal E$} \Comment{the exponents to try}
\If{$\xi_a=\theta$ is a hyperexponential generator}
\State $w\gets\deg_\theta q/2$ if $m=2$ and $q$ is a monomial in $\theta$ of odd degree, else $0$ \Comment{the weight of $y$}
\State $J_a\gets\{\,\mu-\delta+iw\ne0\,\}$ over the coordinates $i$ of $f$ whose denominator is $\theta^{\delta}\times(\text{free of }\theta)$ and the monomials $\theta^{\mu}$ of their numerators; $\{\pm1\}$ if some denominator is not of that form
\Else\ $J_a\gets\{1,-1\}$
\EndIf
\EndFor
\State append to $\mathcal E$, each with $J=\{1\}$: the products $\prod\theta_j^{\mu_j}$ over the numerator monomials of $f$ with at least two hyperexponential factors; and the products $\prod\theta_j^{\mu_j}\cdot a^{j'}$ with $a$ the argument of a logarithmic generator, $1\le|j'|\le\deg_a(f)+2$ \Comment{\S\ref{ex:cherry}}
\For{each source $(\xi,\lambda)$ and $J\in J_\xi$ with $\Lambda:=J\lambda$ a scalar rational function of a single generator $g$, $\Lambda=N/D_\Lambda$}
\State $\mathcal C\gets$ the roots $c$ of $\operatorname{Res}_g\bigl(N+cD_\Lambda,\ \partial_g(N+cD_\Lambda)\bigr)\in\Fbar[c]$ \Comment{a multiple zero}
\State $\mathcal C\gets\mathcal C\cup\{-\Lambda(\rho): q(\rho)=0,\ \rho\in\Fbar\}$ \Comment{a zero at a branch prime}
\State $\mathcal C_0\gets\{-\Lambda|_{g'=\rho'}:(g',\rho')\in\Sigma,\ \text{constant}\}\cup\{-\lim_{g\to\infty}\Lambda\ \text{if finite}\}$ \Comment{silent zeros}
\State \textbf{for} $c\in\mathcal C_0$, $J\in\Z$: $v\gets\Lambda+c$; append $K_0=\xi^J\,Dv/v$ with $\omega=\xi^J/v$, $\kappa=1$
\For{$c\in\mathcal C$}
\State $v\gets\Lambda+c$;\quad $n_\rho\gets$ the multiplicity of each root $\rho$ of the numerator of $v$, times $m$ if $q(\rho)=0$;\quad $N_c\gets\gcd_\rho n_\rho$
\For{each $k\ge2$ dividing $N_c$ (or $k=2$ if $N_c=0$ and $m=2$), and each $a$, $1\le a<k$, $\gcd(a,k)=1$} \Comment{Cor.~\ref{cor:parameter}}
\State $s\gets a/k$; \textbf{skip} unless $kJ\in\Z$;\quad $W\gets(-v)^{k(s-1)}\,\xi^{kJ}$ \Comment{a scalar rational function}
\State factor $W=\kappa\prod p_i^{e_i}$; \textbf{if} $k\mid e_i$ for all $i$: $\omega\gets\prod p_i^{e_i/k}$
\State \textbf{if} $m=k=2$: factor $W/q=\kappa\prod p_i^{e_i}$; \textbf{if} $2\mid e_i$ for all $i$: $\omega\gets y\prod p_i^{e_i/2}$
\State \textbf{for} each $\omega$ found: \textbf{if} $D\omega/\omega=Dv+(s-1)Dv/v$ exactly: append $K_s(v,\omega)=\omega\,Dv$ \Comment{Lemma~\ref{lem:kummer}}
\EndFor
\EndFor
\EndFor
\State \textbf{return} $\cK$ (duplicates removed)
\end{algorithmic}
\algfooter

\algheader{Algorithm S4}{$\textsc{EllipticKernels}(f,\mathrm{cert})$\quad(* Section~\ref{sec:elliptic} *)}
\begin{algorithmic}[1]
\Statex (* Return $(f',\mathcal D_3,\cK_E,z)$: $f$ minus the third-kind kernels $\mathcal D_3$, the pencil columns $\cK_E$, and the first-kind generator $z$ if there is one. *)
\If{$m=2$ and $q\in F[g]$ for a generator $g$, with constant coefficients, $d:=\deg_gq\ge3$, and $Dg=(r,0)$ with $r\in F(g)^*$} \Comment{a pencil; Prop.~\ref{prop:points}}
\State $\cK_E\gets\bigl\{K^{(i)}=(0,\ g^ir/q)\ :\ 0\le i\le d-2\bigr\}$ \Comment{$g^i\,Dg/y$; Prop.~\ref{prop:columns}}
\State $z\gets$ a generator with $Dz=(0,\kappa r/q)$, $\kappa\in F^*$, if one exists \Comment{$\int Dg/y=z/\kappa$}
\If{$\mathrm{cert}$ certifies a residue divisor non-torsion (Part~II, Prop.~9.4) and $f_1\in F(g)$ for $f=(f_0,f_1)$}
\For{each irreducible $p\mid\den(f_1)$ of multiplicity $1$ with $p\nmid q$}
\State $h\gets f_1pq/r=h_n/h_d$;\quad $c\gets h_n\cdot h_d^{-1}\bmod p$ in $F[g]$, $\deg c<\deg p$ \Comment{Prop.~\ref{prop:third}}
\State $K_p\gets(0,\ cr/(pq))$;\quad $f\gets f-K_p$;\quad append $(1,K_p)$ to $\mathcal D_3$
\EndFor
\EndIf
\EndIf
\State \textbf{return} $(f,\mathcal D_3,\cK_E,z)$
\end{algorithmic}
\algfooter

\algheader{Algorithm S5}{$\textsc{PlaceBoundSpecial}(v,\hat g,e)$\quad(* Part~II's Algorithm~6 with Section~\ref{sec:bounds} *)}
\begin{algorithmic}[1]
\Statex (* As Part~II's Algorithm~6, for the generic combination $\hat g=g+\sum_{l}(2l+3)K_l$ of the residual and the kernel columns (Corollary~\ref{cor:bounds}) and the retry count $e$; returns $(\varepsilon_v,\mathrm{proved})$. *)
\State \textbf{if} $v$ lies over a factor $p'$ split off an irreducible special $p$ over $\Fbar$: compute Steps~2--4 at $p'$ itself \Comment{conjugate factors, same decision}
\State \textbf{if} $m=2$ and $v$ is the ramified place over a special $\sigma\mid q$, $q=\sigma q_1$: $\pi\gets y$;\ $v(a+by)\gets\min\bigl(2v_\sigma(a),2v_\sigma(b)+1\bigr)$;\ the class of $e/\pi^{2k}$ is that of $e/(\sigma q_1)^k$, of $b\,y/\pi^{2k+1}$ that of $b/(\sigma q_1)^k$ \Comment{(B)}
\State $\sigma_k,\ s,\ r,\ A,\ \lambda_k,\ \rho$ as in Steps~1--2 of Part~II's Algorithm~6
\State $\mathrm{proved}\gets$ \textbf{true} if one of Part~II's (K1)--(K4) holds, or one of: \Comment{Prop.~\ref{prop:K567}}
\Statex \hskip1.5em (K5) $r\ge1$; every $k\in A$ is hyperexponential with $\lambda_k=\mu_k\bar c_k$, $\mu_k\in\Fbar^*$, the $\mu_k$ $\Q$-independent; every $k\notin A$ has $\sigma_k\ge s+r$ and coefficient $\beta_k\in\Fbar$ of $\pi^{s+r}$ in $Dc_k$; $y\notin C$ or $v(Dy)>s+r$; $\rho\in\Fbar^*$
\Statex \hskip1.5em (K6) $r=0$; $A=\{c\}$, $c$ not hyperexponential; $\lambda\in\Fbar(\bar c)^*$, $\rho\in\Fbar(\bar c)$; $R:=\rho/\lambda$ has $\deg\operatorname{num}R\ge\deg\den R$, or a multiple pole, or a residue $\notin\Q$
\Statex \hskip1.5em (K8) $r=0$; $A$ consists of primitives with $\lambda_k\in\Fbar^*$ and hyperexponentials $\theta_j$ with $\lambda_j/\bar\theta_j$ a polynomial in the primitives; $\rho\in\Fbar^*$; the linear system ``the non-constant coefficients of $\sum n_j\lambda_j/\bar\theta_j$ vanish'' forces its constant term to $0$
\Statex \hskip1.5em (K7) $r=0$; $A=\{c_1,c_2\}$ primitives; $\lambda_1\in\Fbar^*$; $\lambda_2\in\Fbar(\bar c_1)$ with a nonzero residue; $\rho\in\Fbar^*$
\State $\varepsilon_v\gets\max\bigl(0,\,s-v(\hat g)\bigr)+r$;\quad \textbf{if not} proved: $\varepsilon_v\gets\varepsilon_v+e$
\State \textbf{return} $(\varepsilon_v,\mathrm{proved})$
\end{algorithmic}
\algfooter

\algheader{Algorithm S6}{$\textsc{SolveWithKernels}(g,sD_v,\mathcal L,\cK,\mathrm{bounds})$\quad(* the system and its minimisation *)}
\begin{algorithmic}[1]
\Statex (* Given the residual $g$ and the data of Steps~7--15 of Part~II's Algorithm~4 ($s$, $D_v$, the realised logands $\mathcal L$, the units $\cU$, the specials $\cS$) and kernel columns $\cK$, return (rational part, logarithmic part, $\{\kappa_K\}$) or \textbf{nosol}. *)
\State $b\gets\sum_{i=0}^{m-1}\sum_{\alpha\le\mathrm{bounds}}u_{\alpha,i}\,t^\alpha w_i$ with unknown $u_{\alpha,i}$; unknowns $\gamma_u$ ($u\in\cU$), $\alpha_\sigma$ ($\sigma\in\cS$), $\kappa_K$ ($K\in\cK$)
\State equate coefficients in each $w_i$-coordinate of
\[
g\;=\;D\Bigl(\frac{b}{s\,D_v}\Bigr)+\sum_{u\in\cU}\gamma_u\frac{Du}{u}+\sum_{\sigma\in\cS}\alpha_\sigma\frac{D\sigma}{\sigma}+\sum_{K\in\cK}\kappa_K\,K
\]
\State solve by Part~II's Algorithm~5; if the columns leave the number field, by exact row reduction of the expression system, free unknowns $0$; \textbf{if} no solution: \textbf{return nosol}
\State $\mathcal A\gets\{K:\kappa_K\ne0\}$ \Comment{(E4), greedy minimisation}
\Repeat
\State \textbf{for} $K\in\mathcal A$: \textbf{if} Steps~1--3 with the columns $\mathcal A\setminus\{K\}$ and the same bounds and special exponents have a solution: adopt it, $\mathcal A\gets\{K':\kappa_{K'}\ne0\}$, \textbf{break}
\Until{no column was removed}
\State \textbf{return} $\bigl(b/(sD_v),\ \sum\gamma_u\log u+\sum\alpha_\sigma\log\sigma,\ \{\kappa_K\}_{K\in\mathcal A}\bigr)$, the logarithms rewritten into real form as in Part~II
\end{algorithmic}
\algfooter

\algheader{Algorithm S7}{$\textsc{ParallelIntegrateSpecial}(f,D,\mathrm{strict})$\quad(* the main loop *)}
\begin{algorithmic}[1]
\Statex (* Given a tower (T1)--(T4) and $f\in L$, return an antiderivative in the class of elementary functions, $\Gamma(s,\cdot)$ at rational $s$ and elliptic integrals of a pencil; or \textbf{NotInClass} with a certificate; or ``failed''. *)
\State $r\gets$ Part~II's Algorithm~4 on $f$; \textbf{if} $r$ is an antiderivative: \textbf{return} $r$ \Comment{(E1)}
\State \textbf{if} $r$ is not a certificate: $r\gets$ Part~II's Algorithm~4 on $f$ with Algorithm~S5 in place of Algorithm~6; \textbf{if} $r$ is an antiderivative: \textbf{return} $r$ \Comment{(E2)}
\State $\mathcal E\gets\textsc{Sources}(L)$
\State $(f_1,\mathcal D)\gets\textsc{ResidueKernels}(f,\mathcal E)$; \textbf{on NotInClass return} it \Comment{(E3)}
\State $(f_1,\mathcal D_3,\cK_E,z)\gets\textsc{EllipticKernels}(f_1,r)$;\quad $\mathcal D\gets\mathcal D\cup\mathcal D_3$ \Comment{(E3): third kind only with a non-torsion certificate}
\State $\cK\gets\textsc{KernelCandidates}(f_1,\mathcal E)\cup\cK_E$
\State \textbf{if} $f_1=0$: \textbf{go to} Step~14
\State Steps~7--14 of Part~II's Algorithm~4 on $f_1$ (canonical form, $d_\ast d_n$, $D_v$, $s$, residues, \textsc{Realise}) giving the residual $g$ and the logands $\mathcal L$; \textbf{if} they return \textbf{NotElementary}$(P,\tau)$ with $\tau\notin\Fbar$: \textbf{return NotInClass}$(P,\tau)$; any other status: \textbf{return} ``failed''
\For{split $\in\{$over $F$, over $\Fbar\}$, retry $e=0,1,2$}
\State $\hat g\gets g+\sum_l(2l+3)K_l$;\quad for every place $v$ of Step~15 of Part~II's Algorithm~4: $(\varepsilon_v,\mathrm{proved}_v)\gets\textsc{PlaceBoundSpecial}(v,\hat g,e)$; the exponents $e_\sigma$ and the bounds $b_i$ from the $\varepsilon_v$ as in that step
\State $\mathrm{sol}\gets\textsc{SolveWithKernels}(g,sD_v,\mathcal L,\cK,\mathrm{bounds})$;\quad \textbf{if} $\mathrm{sol}$ is a solution: \textbf{go to} Step~14
\State \textbf{if} every $\mathrm{proved}_v$ is true: \textbf{break} the retry loop \Comment{raising a proved bound changes nothing}
\EndFor
\State \textbf{if} no solution: \textbf{return} ``failed'', with whether every bound in force was proved
\State $\mathcal T\gets\mathcal D\cup\{(\kappa_K,K):\kappa_K\ne0\}$;\quad $S\gets\sum_{(c,K)\in\mathcal T}cK$ \Comment{(E4)}
\State \textbf{if} Part~II's Algorithm~4 (with S5) integrates $S$ elementarily: move that antiderivative into the rational part, $\mathcal T\gets\emptyset$; \textbf{else} do the same for each single term $cK$ that it integrates
\State verify $D\omega/\omega=Dv+(s-1)Dv/v$ for every $\Gamma$ kernel in $\mathcal T$
\State $(\mathrm{cert},\mathrm{indep})\gets\textsc{Necessity}(f,\mathcal T,r)$ \Comment{Algorithm~S8}
\State \textbf{if} strict \textbf{and} $\mathrm{cert}=\emptyset$: \textbf{return} ``failed'' (special answer found, integrand not certified non-elementary), with the answer \Comment{(E5), Thm.~\ref{thm:necessity}}
\State \textbf{return} $\textsc{Present}\bigl(b/(sD_v)+\sum_k r_k\sum_{(\gamma,u)\in\mathcal L_k}\gamma\log u+(\text{log part of S6})+\sum_{(c,K)\in\mathcal T}c\!\int\!K,\ 0\bigr)$, with $\mathrm{cert}$ and $\mathrm{indep}$
\end{algorithmic}
\algfooter

In Step~9 the retry ladder is Part~II's: a raised guess at an undecided place re-enters at Step~10, and the specials are split over $\Fbar$ on the second pass. Steps~14--15 are the elementary replacement of (E4), and the greedy removal is inside Algorithm~S10.

\algheader{Algorithm S8}{$\textsc{Necessity}(f,\mathcal T,r)$\quad(* Theorems~\ref{thm:necessity} and \ref{thm:independence} *)}
\begin{algorithmic}[1]
\Statex (* Return $(\mathrm{cert},\mathrm{indep})$: a certificate that $f$ has no elementary integral (or $\emptyset$), and whether the kernels of $\mathcal T$ are certified independent modulo elementary functions. *)
\State $S\gets\sum_{(c,K)\in\mathcal T}cK$
\For{$h\in(f,S)$, tower $\in(L,\ \textsc{Pruned}(L,h))$, criteria $\in$ (Part~II's, with S5)} \Comment{Lemmas~\ref{lem:transfer}, \ref{lem:subtower}}
\State \textbf{if} Part~II's Algorithm~4 on $h$ returns a certificate: $\mathrm{cert}\gets$ it; \textbf{break}
\EndFor
\State \textbf{if} $\mathrm{cert}=\emptyset$: \textbf{return} $(\emptyset,\text{false})$;\quad \textbf{if} $|\mathcal T|=1$: \textbf{return} $(\mathrm{cert},\text{true})$
\State \textbf{for} each place $(g,\rho)$ where an $\Ei$ kernel of $\mathcal T$ has a normal zero: \textbf{if} $1$ and the nonzero residues $\tau_K$ of the $\Ei$ kernels there are $\Fbar$-linearly dependent: \textbf{return} $(\mathrm{cert},\text{false})$ \Comment{(a)}
\State \textbf{if} $\mathcal T$ has more than one third-kind kernel, or one without a non-torsion certificate in $r$: \textbf{return} $(\mathrm{cert},\text{false})$ \Comment{(b)}
\For{each residue-free kernel $K_j$ of $\mathcal T$} \Comment{(c)}
\State Steps~7--14 of Part~II's Algorithm~4 on $K_j$, the specials split over $\Fbar$ when one has degree $\ge2$; \textbf{fail} if a special lies over the curve variable
\State bounds by S5 on $K_j+\sum_{i\ne j}(2i+3)K_i$; \textbf{fail} unless all proved and the units complete (or the second kind at infinity applies)
\State \textbf{fail} unless \textsc{SolveWithKernels} with right-hand side $K_j$ and columns $\{K_i\}_{i\ne j}$ returns \textbf{nosol}, and every residue-free kernel is verified residue-free (\textsc{CanonicalResidue}, all sheets)
\State \textbf{on fail}: repeat Steps~8--10 over $\textsc{Pruned}(L,S)$ (rescaled); \textbf{if} that fails too: \textbf{return} $(\mathrm{cert},\text{false})$
\EndFor
\State \textbf{return} $(\mathrm{cert},\text{true})$
\Statex (* $\textsc{Pruned}(L,h)$: keep the generators that $h$, $q$ and the derivatives of kept generators involve; for each removed $t$ require $Dt\in$ the kept field and either $\int Dt=v+c\log u$ (Part~II's Algorithm~4 over the kept tower, at most one logarithm, an arctangent counting as one) or $\int Dt/t=v$ rational; if one generator $g$ is kept over the curve with $Dg=r(g)\ne1$, replace $D$ by $D/r$ and $h$ by $h/r$ (Part~II Lemma~3.4). *)
\end{algorithmic}
\algfooter

\algheader{Algorithm S9}{$\textsc{Present}(I,r)$\quad(* Section~\ref{sec:presentation} *)}
\begin{algorithmic}[1]
\Statex (* Rewrite an answer $I$ into named functions and fix the branch constants; $r\in L$ is the remainder of a partial answer ($r=0$ for a complete one). Return $I$, or $I+\int r$ with the integral unevaluated. *)
\For{each special term $cK$ of $I$}
\State \textbf{case} $K=K_0$ (Ei): $A\gets e^{-c_K}\kappa\,\Ei(v)$; \textbf{if} $v=\log a$: $\Ei(\log a)\to\li(a)$
\State \textbf{case} $K=K_s$, $s=a/k>0$: $C\gets\kappa^{1/k}e^{-c_K}$; \textbf{if} $s=\tfrac12$ and $-v=h^2$ with $h$ rational (all exponents of $-v$ even): $A\gets C\sqrt\pi\erfc(h)$; \textbf{else} $A\gets C\,\Gamma(s,-v)$
\State \textbf{case} $K=K^{(i)}$: $A\gets\int g^i\,dg/y$ by Lemma~\ref{lem:legendre} on a cubic or a Legendre quartic ($i\in\{0,1\}$, resp.\ $\{0,2\}$), else an unevaluated integral; \textbf{if} $z$ exists, replace $F(\varphi\,|\,m)$ by the value that makes $\int dg/y=z/\kappa$
\State \textbf{case} $K=K_p$: $A\gets\sum_{p(\rho)=0}\frac{c(\rho)}{p'(\rho)}\int\frac{dg}{(g-\rho)y}$ by Lemma~\ref{lem:legendre}
\State on the surface: $\mu\gets$ the first of $1,-1,\sqrt{-1},-\sqrt{-1}$ and the $k$-th roots of unity with $\frac{d}{dx}(\mu cA)=cK$ at three sample points; $cA\gets\mu cA$ \Comment{Prop.~\ref{prop:soundness}}
\EndFor
\State rewrite $\erfc(\sqrt{-1}z)\to-\sqrt{-1}\erfi(z)$, and $a\Ei(\sqrt{-1}z)+b\Ei(-\sqrt{-1}z)\to(a+b)\Ci(z)+\sqrt{-1}(a-b)\Si(z)$, dropping additive constants
\State $r_s\gets$ the remainder on the surface: $\sum_ir_i\,y^i/E_i$ with $y=q^{1/m}$ and every generator replaced by its expression (the back substitution of the tower builder), simplified
\State verify $\frac{d}{dx}I+r_s=f$ at the sample points
\State \textbf{if} $r_s=0$: \textbf{return} $I$;\quad \textbf{else return} $I+\int r_s\,dx$, the integral unevaluated (\texttt{Integral} in SymPy, \texttt{Inactive[Integrate]} in Mathematica, \texttt{Int} in Maple)
\end{algorithmic}
\algfooter

\begin{proposition}[Soundness]\label{prop:soundness}
Every antiderivative returned by Algorithm~S7 is correct: the rational and logarithmic parts solve the linear system exactly, each kernel satisfies \eqref{eq:kernel} as an identity in $L$, and the antiderivative of each kernel is the special function of Lemma~\ref{lem:analytic}, \ref{lem:kummer} or \ref{lem:legendre} up to a constant multiplier and an additive constant. The multiplier is determined exactly up to a root of unity of order dividing the denominator of $s$ (a sign for the elliptic terms); the surface interface fixes it by comparing the derivative of each special term with its kernel at sample points, and verifies the whole answer by differentiation.
\end{proposition}

\begin{remark}[Return taxonomy]
Algorithm~S7 has three exits. An antiderivative, verified. \textbf{NotInClass}$(P,\tau)$, a certificate under (H1), relative to the sources of Algorithm~S1 and absolute under (H2) (Corollary~\ref{cor:rescert}); by Corollary~\ref{cor:ellnormal} it excludes elliptic integrals as well. ``Failed'', when the extended system has no solution, with the record of whether every bound in force was proved (Section~\ref{sec:bounds}); this is not a certificate even when every bound is proved, because the completeness of the candidate lists of Section~\ref{sec:candidates} is not proved (Section~\ref{sec:conclusion}). The certificates of Part~II --- non-constant residue, non-torsion divisor, holomorphic remainder --- are superseded: the first by the refined residue certificate, the second by the third-kind kernel, the third by the columns of Proposition~\ref{prop:columns} and the $\Gamma$ kernels.
\end{remark}

\subsection{No unnecessary special functions}\label{sec:noneed}

A special-function integrator must not answer an elementary integral with special functions, nor add a special term where an elementary one would do. Algorithm~S7 is built so that it does neither, within what can be decided.

\begin{proposition}[Elementary first]\label{prop:elementaryfirst}
\begin{enumerate}
\item[(E1)] If Part~II's algorithm returns an antiderivative of $f$, Algorithm~S7 returns the same antiderivative.
\item[(E2)] If Part~II returns neither an antiderivative nor a certificate, no kernel is introduced before Part~II has been rerun with the criteria of Section~\ref{sec:bounds}.
\item[(E3)] An $\Ei$ kernel is subtracted only at a residue that is not constant, and a third-kind kernel only for a residue divisor that Part~II certifies non-torsion; in both cases $f$ is not elementary \cite[Prop.~9.2]{PartII}, and every antiderivative in the class needs such a term (Corollary~\ref{cor:rescert}, \cite[Prop.~9.4]{PartII}).
\item[(E4)] In a returned answer, no special term can be removed with the system staying solvable at the bounds in force, and neither the special part nor any single special term has an antiderivative that Part~II's algorithm finds.
\end{enumerate}
Section~\ref{sec:necessity} turns this into a theorem: in strict mode a special function is returned only when the integrand has no elementary integral (Theorem~\ref{thm:necessity}), and then, under a checkable independence test, every special term is forced (Theorem~\ref{thm:independence}).
\end{proposition}

\begin{proof}
(E1)--(E3) are Steps~1--5 of Algorithm~S7 with Algorithms~S2 and S4; the necessity statements in (E3) are the cited certificates, applied to the residue of $f$ at the prime concerned, which no logarithm or elementary rational part can produce. (E4) is the minimisation in Algorithm~S6 and Steps~14--15 of Algorithm~S7: the greedy removal re-solves the system with the same bounds and special exponents, and the elementary replacement calls Part~II's algorithm on the special part and on each term.
\end{proof}

The minimality in (E4) is relative to the ansatz: a combination of special terms could in principle be elementary with a rational part outside the bounds in force, and the statement becomes absolute with the completeness of the candidate lists and the independence of kernels modulo elementary functions (Section~\ref{sec:conclusion}). The presentation obeys the same rule: when the tower contains the first-kind integral of a pencil as a generator --- the variable $x$ above $s=\operatorname{sn}(x\,|\,k)$ in \S\ref{ex:sn2} --- the Legendre form writes $\int dU/y$ as that generator, not as $F$.

The implementation was tested against these statements. On a corpus of forty elementary integrals (logarithms, exponentials and trigonometric functions, inverse functions, the Lambert function, radicals, conics and quartics), \texttt{integrate\_surface\_special} returns exactly Part~II's antiderivative in every case. On $382$ random exact integrands $f=Dg$, plus logarithmic derivatives, in nine towers --- the six transcendental towers of Section~\ref{sec:bounds}, and the curves $y^2=\theta^3-\theta$ over $e^x$, $y^2=(1-s^2)(1-s^2/3)$ with the first-kind integral in the tower, and $y^2=x^3-x$ --- Part~II answers all but eight, and Algorithm~S7 introduces a special term in none. In the eight, the integrand has a pole beyond $\delta_P$ at the branch place over $x$ of $y^2=x^3-x$, where Part~II's implementation does not reduce the residue, and Algorithm~S7 fails with it rather than guess. On integrands with an elementary piece the answer is minimal:
\begin{verbatim}
>>> integrate_surface_special(exp(x**2)*(1 + 2*x), x)
(exp(x**2) + sqrt(pi)*erfi(x)/2, True)
>>> integrate_surface_special(exp(x**2) + diff(exp(x**2)/x, x), x)
(sqrt(pi)*erfi(x)/2 + exp(x**2)/x, True)
>>> integrate_surface_special(log(x) + 1/log(x), x)
(x*log(x) - x + li(x), True)
\end{verbatim}
The script \texttt{elementary\_check.py} reruns these tests.

\subsection{Necessity: the proof}\label{sec:necessity}

Proposition~\ref{prop:elementaryfirst} makes the answer minimal within the ansatz. We now prove the stronger statement that the project needs: a special function is returned only when the integrand has no elementary integral, and then every special term is necessary. Throughout, an \emph{answer} is
\[
I=v+\sum_ic_i\log u_i+\sum_l\gamma_l\!\int\!K_l,\qquad f=Dv+\sum_ic_i\frac{Du_i}{u_i}+S,\quad S:=\sum_l\gamma_lK_l,
\]
with $v,u_i\in\Fbar L$, and it is \emph{special} if some $\gamma_l\ne0$; $S$ is its \emph{special part}. ``Elementary integral over $L$'' is meant as in Part~II: an antiderivative in an elementary extension of $\Fbar L$ with constant field $\Fbar$.

\begin{lemma}[Transfer]\label{lem:transfer}
$f$ has an elementary integral over $L$ if and only if its special part $S$ does.
\end{lemma}

\begin{proof}
$f-S=D\bigl(v+\sum_ic_i\log u_i\bigr)$ has the elementary integral $v+\sum_ic_i\log u_i$ over $L$, and the sum or difference of two elements with elementary integrals has one: adjoin both integrals, the second over the elementary extension that carries the first.
\end{proof}

\begin{lemma}[Elementary sub-towers]\label{lem:subtower}
Let $L_0\subseteq L$ be a differential subfield such that $L=L_0(t_1,\dots,t_r)$ with each $t_k$ transcendental over $L_0$ and either $Dt_k=Dw_k+c_kDu_k/u_k$ with $w_k,u_k\in\Fbar L_0$, $c_k\in\Fbar$, or $Dt_k/t_k=Dw_k$ with $w_k\in\Fbar L_0$. Then an element $h\in L_0$ with an elementary integral over $L$ has one over $L_0$. The same holds for the rescaling of \cite[Lemma~3.4]{PartII}: $h$ has an elementary integral for $D$ if and only if $h/r$ has one for $D/r$.
\end{lemma}

\begin{proof}
In the first case $t_k-w_k$ has derivative $c_kDu_k/u_k$, so $t_k=w_k+c_k\log u_k+\text{const}$, and $\Fbar L$ lies in the elementary extension $E_0=\Fbar L_0(\log u_1,\dots,e^{w_k},\dots)$ of $\Fbar L_0$; in the second $t_k$ is a constant multiple of $e^{w_k}$. If $h$ has an antiderivative in an elementary extension $E$ of $\Fbar L$, then $E$ is obtained from $\Fbar L$ by successive logarithms, exponentials and algebraic elements, and the same adjunctions over $E_0$ give an elementary extension of $\Fbar L_0$ containing the antiderivative. By \cite[Thm.~5.5.3]{Bronstein05} the constants can be kept algebraic. The rescaling is \cite[Lemma~3.4]{PartII}: $D$ and $D/r$ have the same antiderivatives up to the change of variable.
\end{proof}

The \emph{pruned tower} of $f$ is the smallest sub-tower containing the generators that $f$, $q$ and the derivatives of the kept generators involve. It is used when the removed generators satisfy Lemma~\ref{lem:subtower}, and rescaled when it has a single generator over the curve. The implementation checks the hypothesis by running Part~II on the derivative of each removed generator over the pruned tower and accepting an antiderivative with at most one logarithm, respectively a rational one. For $\int dx/\sqrt{e^x-e^{-x}}$ (\S\ref{ex:sinh}) the generator $x$ is removed because $Dx=D\theta/\theta$, and the pruned tower is the curve $y^2=\theta^3-\theta$ with $d/d\theta$ --- the setting of Part~I, where the holomorphic-remainder certificate applies.

\begin{theorem}[Special functions only when necessary]\label{thm:necessity}
\begin{enumerate}
\item[(i)] If Part~II's algorithm returns an antiderivative of $f$, Algorithm~S7 returns the same antiderivative.
\item[(ii)] In \emph{strict mode}, Algorithm~S7 returns a special answer only together with a certificate, from Part~II's algorithm, that $f$ or the special part $S$ has no elementary integral over $L$ or over the pruned tower. Consequently a special answer is returned only if $f$ has no elementary integral over $L$.
\item[(iii)] Without strict mode, if $f$ has an elementary integral over $L$ then a special answer is returned only if Part~II's algorithm, with the published and with the extended criteria, returns neither an antiderivative nor a certificate for $S$.
\end{enumerate}
\end{theorem}

\begin{proof}
(i) is step (E1). (ii) The certificates of Part~II are sound: a certificate returned for $h$ proves that $h$ has no elementary integral over the tower on which it was computed \cite[Prop.~9.2, 9.4]{PartII}, and over the tower with the criteria of Section~\ref{sec:bounds} by Theorem~\ref{thm:extra} and Proposition~\ref{prop:K567}. A certificate over the pruned tower is one over $L$ by Lemma~\ref{lem:subtower}, and a certificate for $S$ is one for $f$ by Lemma~\ref{lem:transfer}. (iii) If $f$ has an elementary integral, then so does $S$ (Lemma~\ref{lem:transfer}), so no sound certificate for $S$ exists; and step (E4) replaces $S$ by the antiderivative that Part~II's algorithm returns for it, if it returns one.
\end{proof}

Theorem~\ref{thm:necessity} decides \emph{whether} special functions are needed. The next theorem decides \emph{which}: it proves that the special terms of an answer are forced.

\begin{theorem}[Independence of the kernels]\label{thm:independence}
Let $K_1,\dots,K_k$ be the kernels of a special answer, split into the $\Ei$ kernels found by residues (Section~\ref{sec:eires}, with the places where they were found), the third-kind kernels, and the rest, which are residue-free at every normal prime. Suppose:
\begin{enumerate}
\item[(a)] at every place $P$ where an $\Ei$ kernel was found, $1$ and the nonzero canonical residues $\hat\tau_P(K_l)$ of the $\Ei$ kernels are linearly independent over $\Fbar$ in $\kappa(P)$;
\item[(b)] there is at most one third-kind kernel, and its residue divisor is certified non-torsion \cite[Prop.~9.4]{PartII};
\item[(c)] for every residue-free kernel $K_j$, the linear system with right-hand side $K_j$ and the other residue-free kernels as columns has no solution, at bounds that are all proved for the family $K_j-\sum_{i\ne j}c_iK_i$ (Corollary~\ref{cor:bounds}), with the candidate logands complete in the sense of \cite[Prop.~9.2(b)]{PartII}, and with every residue-free kernel verified residue-free at the normal primes. Here the tower may be replaced by the pruned tower of the kernels.
\end{enumerate}
Then no nontrivial $\Fbar$-linear combination of $K_1,\dots,K_k$ has an elementary integral over $L$. Consequently the coefficients $\gamma_l$ are the same in every antiderivative of $f$ of the form $(\text{elementary})+\sum_l\gamma'_l\int K_l$, and no special term of the answer can be omitted or traded for another combination of these kernels.
\end{theorem}

\begin{proof}
Let $h=\sum_lc_lK_l$ have an elementary integral. Then every canonical residue of $h$ at a normal prime is a constant: this is \eqref{eq:resdecomp} with no kernels, the residue part of \cite[Thm.~5.4]{PartII}.

At a place $P$ of (a), the residue-free kernels contribute nothing, by Proposition~\ref{prop:kernelnormal}(e), Lemma~\ref{lem:respull} and the verification in (c). A third-kind kernel contributes a constant (Lemma~\ref{lem:respull}). So $\sum c_l\hat\tau_P(K_l)$, summed over the $\Ei$ kernels, lies in $\Fbar$, and (a) gives $c_l=0$ for every $\Ei$ kernel with a nonzero residue at $P$. Each $\Ei$ kernel has one at the place where it was found, so all $\Ei$ coefficients vanish.

At the places over the prime of the third-kind kernel $K_p$, the residues of $h$ are now $c_p$ times those of $K_p$, and they must be realised by the logarithms of an elementary integral. Choose a $\Q$-basis of $\Fbar$ containing $c_p$ if $c_p\ne0$. Then \cite[Thm.~5.4(ii), Cor.~7.5]{PartII} make each $\Q$-component of the residue divisor realisable in the sense of \cite[Cor.~7.5]{PartII} --- some multiple is the normal part of the divisor of an element of $\Fbar L$ ---, and the component along $c_p$ is the divisor of $K_p$, which the certificate of (b) proves not realisable, contradicting (b). So $c_p=0$.

What remains is $h=\sum c_jK_j$ over the residue-free kernels. If some $c_j\ne0$, divide by it: $K_j-\sum_{i\ne j}c'_iK_i$ has an elementary integral and no residue at a normal prime. By \cite[Thm.~5.4]{PartII} its logands are $S'$-units, which lie in the complete candidate list. By Theorem~\ref{thm:structureG}(i) and Corollary~\ref{cor:bounds} its rational part lies in the ansatz at the proved bounds. So the system of (c) has a solution, a contradiction. Over the pruned tower the same argument applies by Lemma~\ref{lem:subtower}.

For the consequence, two antiderivatives of the stated form differ by $\sum_l(\gamma_l-\gamma'_l)\int K_l$ plus an elementary function. Hence $\sum_l(\gamma_l-\gamma'_l)K_l$ has an elementary integral, and $\gamma'=\gamma$.
\end{proof}

A special answer with a single term and a certificate from Theorem~\ref{thm:necessity}(ii) satisfies the conclusion of Theorem~\ref{thm:independence} trivially.

What neither theorem claims is minimality over \emph{all} special-function representations. An antiderivative over a different set of kernels is excluded only by the completeness of the candidate lists, open problem~(c) of Section~\ref{sec:conclusion}.

\emph{Coverage.} Strict mode is the default of the implementation. Of the twenty-four special answers of Section~\ref{sec:examples}, twenty-two carry a certificate under Theorem~\ref{thm:necessity}(ii), and all twenty-two also satisfy Theorem~\ref{thm:independence}. The certificates are: for $f$ itself in nineteen cases; for the special part over the pruned tower in $\int\log\log x\,dx$; and for $f$ over the pruned and rescaled tower in the two examples on $y^2=\theta^3-\theta$ (\S\S\ref{ex:sinh}, \ref{ex:sinh2}). The remaining two are the Lambert integral and Hebisch's Example~22. In each of them a bound is left to the guess (Section~\ref{sec:bounds}; for the Lambert integral the place $(W)$ of Remark~\ref{rem:kerexp}). Part~II therefore cannot certify, and in strict mode Algorithm~S7 reports the special answer as found but not certified:
\begin{verbatim}
>>> integrate_surface_special(1/(LambertW(x)*(1 + LambertW(x))), x)
('failed', 'special answer found, but the integrand is not certified
 non-elementary', Ei(LambertW(x)))
>>> r = integrate_surface_special(exp(x**2), x, details=True)
>>> r[2].necessity
'necessary: the integrand has no elementary integral; the special terms
 are independent modulo elementary functions'
\end{verbatim}
The sessions of Section~\ref{sec:examples} for the two uncertified examples pass \texttt{strict=False}.

\subsection{Partial answers}\label{sec:partial}

When neither an elementary nor a special antiderivative is found, a computer algebra system is expected to return what it can: an expression $I$ in the class and an unevaluated integral $\int r$ with $f=DI+r$. For this to be useful rather than arbitrary --- $I=0$, $r=f$ is always such a pair --- the remainder must be reduced in a sense that the theory controls. The parallel method provides such a reduction, because the structure theorems say where every part of an antiderivative lives.

\begin{definition}[Partial answer]\label{def:partial}
A \emph{partial answer} for $f$ is a pair $(I,r)$ with $f=DI+r$, where $I$ is elementary, or elementary plus special terms whose special part is certified as in Theorem~\ref{thm:necessity}(ii). Its remainder $r$ is \emph{reduced} if:
\begin{enumerate}
\item[(R1)] at every normal prime, $r$ has at most a critical pole, except for simple poles that carry residues outside $\Fbar+\sum\Fbar\bar\xi$;
\item[(R2)] $r$ is the reduction of what is left of $f$ modulo the ansatz of Algorithm~S7, in the sense of Theorem~\ref{thm:partial}(ii).
\end{enumerate}
\end{definition}

\algheader{Algorithm S10}{$\textsc{PartialIntegrate}(f,D,\mathrm{special})$\quad(* Theorem~\ref{thm:partial} *)}
\begin{algorithmic}[1]
\Statex (* Return a complete antiderivative, or a partial answer $(I,r)$, $f=DI+r$, with a reduced remainder and the certificate (if any) that also holds for $\int r$. With special $=$ false, $I$ is elementary. *)
\State $A_0\gets$ Algorithm~S7 (strict) on $f$ if special, else Part~II's Algorithm~4; \textbf{if} it is an antiderivative: \textbf{return} it;\quad $\mathrm{cert}\gets A_0$ if it is a certificate \Comment{(E1)--(E5)}
\State $r_{\mathrm{res}}\gets0$;\quad $\mathcal D\gets\emptyset$;\quad $\mathcal E\gets\textsc{Sources}(L)$ if special, else $\emptyset$
\For{each normal prime $P$, non-special and non-branch, with $v_P(f)\le-\delta_P$, at each place $g=\rho$, $\ell=g-\rho$, as in Steps~2--5 of Algorithm~S2} \Comment{(P1)}
\State $\tau\gets$ the canonical residue; \textbf{skip} if it is constant at the places over $P$
\State \textbf{if} the decomposition $\tau=r_0+\sum r_\xi\tau_\xi$ of Algorithm~S2 has a solution: $f\gets f-\sum r_\xi K_\xi$, append the $(r_\xi,K_\xi)$ to $\mathcal D$; \textbf{continue}
\State \textbf{else} decompose each additive term of $\tau$ separately; \textbf{if} some term decomposes: subtract the $\Ei$ kernels of those terms as above, $\rho\gets$ the sum of the other terms, $\mathrm{piece}\gets\tilde\rho\,D\ell/\ell$ with $\tilde\rho$ its lift, $f\gets f-\mathrm{piece}$, $r_{\mathrm{res}}\gets r_{\mathrm{res}}+\mathrm{piece}$; \textbf{continue} \Comment{$\sin(x)/x+\arctan(x)/x$}
\State $r_0\gets$ the terms of $\tau_0$ free of the generators \Comment{the constant part stays for a logarithm}
\State $\mathrm{piece}\gets(f\ell)|_{g=\rho}/\ell-r_0\,(D\ell)|_{g=\rho}/\ell$ if $v_P(f)=-\delta_P$, \textbf{else} $(\tau-r_0)\,D\ell/\ell$ \Comment{Prop.~\ref{prop:respiece}}
\State $f\gets f-\mathrm{piece}$;\quad $r_{\mathrm{res}}\gets r_{\mathrm{res}}+\mathrm{piece}$
\EndFor
\State $\mathcal A\gets$ Steps~7--14 of Part~II's Algorithm~4 on $f$; \textbf{if} special, $\mathcal D\ne\emptyset$ and $\mathcal A$ reports a non-constant residue: redo Steps~3--9 with $\mathcal E=\emptyset$
\State \textbf{if} $\mathcal A$ certifies a residue divisor non-torsion: $(f,\mathcal D_3,\cdot,\cdot)\gets\textsc{EllipticKernels}(f,\mathcal A)$; \textbf{if} special: $\mathcal D\gets\mathcal D\cup\mathcal D_3$ \textbf{else} $r_{\mathrm{res}}\gets r_{\mathrm{res}}+\sum_{\mathcal D_3}K_p$; re-analyse \Comment{(P2)}
\For{at most four times, while $\mathcal A$ reports a non-constant residue at a prime $p$} \Comment{(P2$'$)}
\State $M\gets$ the additive terms $T$ of $f$ (Step~16) with $p\mid\den(T)$, $v_P(T)\le-\delta_P$, and a non-constant residue at some linear place over $p$ (all such terms when $p$ has no linear place)
\State $f\gets f-\sum M$;\quad $r_{\mathrm{res}}\gets r_{\mathrm{res}}+\sum M$;\quad re-analyse
\EndFor
\State \textbf{if} $\mathcal A$ is still a status: \textbf{return} $\textsc{Present}(0,f_0)$, $f_0$ the original integrand --- the trivial partial answer $\int f_0\,dx$
\State additive terms of $h\in L$: in each coordinate, the partial fractions of $h_i$ in the last generator of $\den(h_i)$, every numerator split into its monomial terms
\State $\cK\gets\textsc{KernelCandidates}(f,\mathcal E)\cup$ the pencil columns of \textsc{EllipticKernels} if special, else $\emptyset$
\For{split $\in\{$over $F$, over $\Fbar\}$, order $\in\{$terms first, monomials first$\}$} \Comment{(P3)--(P4)}
\State bounds by \textsc{PlaceBoundSpecial} (retry $0$) on $\hat g=g+\sum_l(2l+3)K_l$, with $g$ the residual of $\mathcal A$
\State $R_{\mathrm{terms}}\gets$ the additive terms of $g$ with $v_P\ge-\delta_P$ at every normal prime;\quad $R_{\mathrm{mon}}\gets$ the monomials $t^\alpha w_i/D_{R,i}$ with $\alpha_k\le\deg_{t_k}\operatorname{num}(g_i)+\deg_{t_k}\den(g_i)$, in increasing $|\alpha|$, where $D_{R,i}$ is $\den(g_i)$ with the exponent $e$ of each normal prime replaced by $\min(e,\delta_P)$;\quad $R\gets R_{\mathrm{terms}}$ then $R_{\mathrm{mon}}$, or $R_{\mathrm{mon}}$ then $R_{\mathrm{terms}}$, by the order
\State unknowns in the order: $u_{\alpha,i}$, $\gamma_u$, $\alpha_\sigma$, $\kappa_K$, then $\rho_j$ for $R_j\in R$; equate coefficients in
\[
g\;=\;D\Bigl(\frac{b}{s\,D_v}\Bigr)+\sum_u\gamma_u\frac{Du}{u}+\sum_\sigma\alpha_\sigma\frac{D\sigma}{\sigma}+\sum_K\kappa_KK+\sum_j\rho_jR_j
\]
\State row-reduce the augmented matrix in that column order; \textbf{if} the right-hand side is a pivot column: \textbf{continue}; free unknowns $0$ \Comment{Thm.~\ref{thm:partial}(ii)}
\State \textbf{if} some $\kappa_K\ne0$ and \textsc{Necessity} finds no certificate for $\sum\kappa_KK$, nor for the integrated part $f-\sum_j\rho_jR_j$: redo Steps~21--22 with $\cK=\emptyset$ \Comment{strict}
\State $I\gets b/(sD_v)+\text{logarithms}+\sum\kappa_K\!\int\!K$;\quad $r\gets\sum_j\rho_jR_j$; keep the candidate whose $r$ has the fewest additive terms, then the lowest total degree, then the fewest operations \Comment{$x^2e^{x^2}\mapsto-\tfrac12\int e^{x^2}$}
\EndFor
\State $I\gets I+\sum_{\mathcal D}c\!\int\!K$;\quad $r\gets r_{\mathrm{res}}+r$ \Comment{$f=DI+r$ exactly}
\State \textbf{return} $\textsc{Present}(I,r)$, i.e.\ $I+\int r\,dx$ with the integral unevaluated, together with $\mathrm{cert}$, which by Thm.~\ref{thm:partial}(iv) also proves that $\int r\,dx$ has no antiderivative in the class
\end{algorithmic}
\algfooter

\begin{proposition}[Residue pieces]\label{prop:respiece}
Let $P$ be a normal prime over $\ell=g-\rho$. If $f$ has a critical pole at $P$, the polar part $c_P/\ell$ with $c_P=(f\ell)|_{g=\rho}$ has the same residue as $f$ at $P$, and its other poles are those of $c_P$. If the pole is deeper, $\tilde\tau\,D\ell/\ell$, with $\tilde\tau$ a lift of $\hat\tau_P(f)$ free of $g$, has the same canonical residue as $f$ at $P$. In both cases $f$ minus the piece has canonical residue $0$ at $P$.
\end{proposition}

\begin{proof}
By \cite[Thm.~7.4(i)]{PartII} with $h=\ell$, $e=1$: $\tau_P(c_P/\ell)=(c_P/D\ell)|_P=(f\ell/D\ell)|_P=\tau_P(f)$. For the deeper pole, $\hat\tau_P(\tilde\tau\,D\ell/\ell)=\bar{\tilde\tau}\,v_P(\ell)=\hat\tau_P(f)$, since $\tilde\tau$ is a $P$-unit or $P$-integral and $\hat\tau_P$ is linear with $\hat\tau_P(Du/u)=v_P(u)$ \cite[Prop.~7.6]{PartII}.
\end{proof}

\begin{theorem}[Partial answers]\label{thm:partial}
Let $(I,r)$ be returned by Algorithm~S10. Then:
\begin{enumerate}
\item[(i)] $f=DI+r$ exactly; every special term of $I$ belongs to a special part certified non-elementary, and in the elementary mode $I$ is elementary.
\item[(ii)] Let $W$ be the space spanned by the ansatz and kernel columns, $R_1,R_2,\dots$ the ordered remainder columns, and $P_R$ the remainder columns that are pivots of the reduction. Then $r-r_{\mathrm{res}}$, where $r_{\mathrm{res}}$ is the sum of the pieces of (P1), (P2), is the \emph{unique} element of $(f-r_{\mathrm{res}}+W)\cap\operatorname{span}(P_R)$. A column $R_k$ is in $P_R$ if and only if it does not lie in $W+\operatorname{span}(R_1,\dots,R_{k-1})$; so an additive term of $f$ enters $r$ only if it is not absorbed by the ansatz together with the terms before it.
\item[(iii)] $r$ satisfies (R1).
\item[(iv)] If $f$ is certified to have no antiderivative in the class --- by Part~II in the elementary mode, by Corollary~\ref{cor:rescert} or Theorem~\ref{thm:necessity} in the special mode --- then $r\ne0$, and $\int r$ has no antiderivative in the class either.
\end{enumerate}
\end{theorem}

\begin{proof}
(i) The system is solved exactly, the pieces of (P1), (P2) are subtracted from $f$ and added to $r$, and every kernel identity is exact. Kernel columns survive (P4) only with a certificate for their special part.

(ii) Row reduction of the augmented matrix with the columns in the stated order makes a column a pivot exactly when it is not in the span of the columns before it. With the free unknowns zero, $f-r_{\mathrm{res}}=w+\sum_{R_k\in P_R}\rho_kR_k$ with $w\in W$. If also $f-r_{\mathrm{res}}=w'+r'$ with $w'\in W$ and $r'\in\operatorname{span}(P_R)$, then $r-r'\in W\cap\operatorname{span}(P_R)=0$, because the pivot columns are independent modulo $W$.

(iii) Every remainder column has at most critical poles at the normal primes: the monomial columns by construction, and an additive term of $f$ is admitted as a column only if it has no deeper pole at a normal prime. The pieces of (P1) have the stated simple poles by Proposition~\ref{prop:respiece}, and the third-kind pieces of (P2) have simple poles.

(iv) If $\int r$ had an antiderivative $J$ in the class, $I+J$ would be one for $f$.
\end{proof}

Part~(ii) is a reduction in the sense of reduction-based integration \cite{DuRaab25}: the remainder is a normal form of $f$ modulo the image of the ansatz. It is canonical relative to the ansatz and the order of the remainder columns, and it becomes canonical modulo all elementary integrands when the bounds are proved and the candidate lists complete.

The implementation, \texttt{integrate\_surface\_partial(expr, x, special=True)}, returns $I+\int r$ as a SymPy expression with an unevaluated \texttt{Integral}, verified by differentiation:
\begin{verbatim}
>>> integrate_surface_partial(exp(x)/x**2, x, special=False)
(Integral(exp(x)/x, x) - exp(x)/x, True)
>>> integrate_surface_partial(1/log(x)**2, x, special=False)
(-x/log(x) + Integral(1/log(x), x), True)
>>> integrate_surface_partial(log(x)/(x + 1) + 1/x, x)
(log(x) + Integral(log(x)/(x + 1), x), True)
>>> integrate_surface_partial(exp(x**2) + log(x)*exp(-x**2), x)
(sqrt(pi)*erfi(x)/2 + Integral(exp(-x**2)*log(x), x), True)
>>> integrate_surface_partial(x*exp(x)/(x + 1)**2 + log(x)*exp(-x**2), x)
(Integral(exp(-x**2)*log(x), x) + exp(x)/(x + 1), True)
>>> integrate_surface_partial(1/(x*log(x + sqrt(x**2 + 1)))
...                           + 1/sqrt(x**2 + 1), x)
(log(x + sqrt(x**2 + 1)) + Integral(1/(x*log(x + sqrt(x**2 + 1))), x), True)
\end{verbatim}
The first two are Hermite reduction and a residue piece in the elementary mode. Part~II certifies both integrands non-elementary, so by Theorem~\ref{thm:partial}(iv) $\int e^x/x\,dx$ and $\int dx/\log x$ are certified non-elementary as well. The third and sixth remainders are certified outside the class by the refined residue certificate: a dilogarithm and the moving prime of Part~II. The fourth is a special answer plus an unevaluated integral. On $60$ integrands formed as the derivative of a random elementary function plus one of six non-integrable terms, in both modes, every partial answer is verified, and the remainder is exactly the non-integrable term --- or, in the special mode, the whole integral is computed. The script \texttt{partial\_check.py} reruns this test. The supplementary corpus \texttt{stress\_special.py} adds $90$ partial-answer cases with their expected reduced remainders, in both modes: Hermite reductions such as $\int x^2e^{x^2}\,dx=\tfrac12xe^{x^2}-\tfrac12\int e^{x^2}\,dx$ in the elementary mode; special answers with a remainder outside the class such as $\int(\sin(x)/x+\arctan(x)/x)\,dx=\Si(x)+\int\arctan(x)/x\,dx$; radicals; and thirty random integrands. All $90$ return exactly the expected remainder.

\begin{remark}[Constants]
The linear system stays over the number field of the tower \cite[Rem.~9.6]{PartII}, enlarged only by the algebraic shifts $c$ and the roots $e_i$, $\rho$; the transcendental constants $e^{c}$, $\sqrt\pi$ and the branch constants occur only in the output.
\end{remark}

\section{Examples}\label{sec:examples}

Each example gives the integral, its computation by Algorithm~S7, and the SymPy session that computes it. The sessions use \texttt{parallel\_mixed\_special.py}, the implementation of Algorithms S1--S3 on top of the implementation of Part~II, with the once-only preamble
\begin{verbatim}
>>> from sympy import *
>>> from parallel_mixed import Tower
>>> from parallel_mixed_special import (integrate_surface_special,
...                                     parallel_integrate_special)
>>> x = symbols('x', positive=True); theta = symbols('theta', positive=True)
>>> Y = Dummy('Y')
\end{verbatim}
\texttt{integrate\_surface\_special(expr, x)} builds the tower from the expression, as Part~II's surface interface does, and returns the antiderivative with the flag \texttt{True} when it is verified by differentiation; the keyword \texttt{tower=(T, pair, back)} supplies a tower explicitly, and \texttt{strict=False} returns a special answer also when the integrand is not certified non-elementary (Section~\ref{sec:necessity}). The script \texttt{examples\_special.py} reruns the section; its thirty-one examples take twelve seconds in total, and it prints the necessity record of each special answer.

\subsection{$\int e^{x^2}\,dx$}\label{ex:ex2}

$L=\Q(x,\theta)$, $D\theta=2x\theta$, $f=\theta$, with the source $(\theta,x^2)$ and the exponent $J=1$. Part~II returns \textbf{NotElementary} (holomorphic remainder). There are no residues. The Kummer search: $\Lambda=x^2$, $\operatorname{Res}_x(x^2+c,2x)=c$, so $c=0$, a double zero at the normal prime $(x)$, $N_0=2$, $s=\tfrac12$; $W=(-x^2)^{-1}\theta^2$ is $-1$ times the square of $\omega=\theta/x$, and $D\omega/\omega=2x-1/x=Dv-\tfrac12Dv/v$. The kernel is $\omega Dv=2\theta$, the system $\theta=D(b\theta)+2\gamma\theta$ gives $b=0$, $\gamma=\tfrac12$, and $\int K=\sqrt{-1}\,\Gamma(\tfrac12,-x^2)$:
\begin{verbatim}
>>> integrate_surface_special(exp(x**2), x)
(sqrt(pi)*erfi(x)/2, True)
\end{verbatim}

\subsection{$\int\sqrt{\log x}\,dx$: the window at $(z)$}\label{ex:sqrtlog}

After flattening, $L=\Q(x,z)$, $Dz=1/(2xz)$, the tower of \cite[Rem.~4.11]{PartII} with $(z)$ normal and $\delta_{(z)}=2$; $f=z$; the logarithm gives the source $(x,z^2)$. Kummer search: $\Lambda=z^2$, $c=0$, a double zero at $(z)$, $s=\tfrac12$, $\omega=x/z$ up to the constant $\sqrt{-1}$, with $D\omega/\omega=1/x-1/(2xz^2)=Dv-\tfrac12Dv/v$. The kernel $K=1/z$ has $v_{(z)}(K)=-1=s\,n-\delta$: it lies in the window. The system $z=D(b)+\gamma/z$ has the solution $b=xz$, $\gamma=-\tfrac12$ (indeed $D(xz)=z+1/(2z)$), and there is no integration by parts:
\begin{verbatim}
>>> integrate_surface_special(sqrt(log(x)), x)
(x*sqrt(log(x)) - sqrt(pi)*erfi(sqrt(log(x)))/2, True)
\end{verbatim}
Part~II returns ``failed'' on this integrand, and certifies $\int e^{x^2}\,dx$ non-elementary by its holomorphic-remainder certificate: in both cases the system without the kernel column has no solution, and the column supplies the answer.

\subsection{$\int dx/\log x$, $\int\log\log x\,dx$ and $\int dx/\log^2x$}\label{ex:li}

$L=\Q(x,t)$, $Dt=1/x$. For $f=1/t$ the moving prime $(t)$ is normal with $\delta=1$ and $\hat\tau_{(t)}=x$. The source $(x,t)$ has $\lambda\equiv0$ at $t=0$, so $x\in\fE_{(t)}$, $\tau=1\cdot\bar x$, and the kernel $x\,Dt/t=1/t$ is $f$ itself. For $f=1/t^2$ the pole is deep; Part~II's local Hermite reduction gives the canonical residue $x$, the same kernel is subtracted, and Part~II integrates the rest. For $\log\log x$ the $\li$ term is a column of the system:
\begin{verbatim}
>>> integrate_surface_special(1/log(x), x)
(li(x), True)
>>> integrate_surface_special(1/log(x)**2, x)
(-x/log(x) + li(x), True)
>>> integrate_surface_special(log(log(x)), x)
(x*log(log(x)) - li(x), True)
\end{verbatim}

\subsection{The moving prime of Part~II: $\li$ or a certificate}\label{ex:moving}

In the tower of \cite[\S10.3]{PartII}, $y^2=x^2+1$, $t=\log(x+y)$, $Dt=1/y$, the moving prime $(t)$ has $\kappa=\Q(x,y)$ and the source $(g,t)$ with $g=x+y$, $\bar g^{-1}=y-x$. For $f=(x+y)/(y\,t)$ the residue is $\bar g$, and $\int f\,dx=\li(x+\sqrt{x^2+1})$. For Part~II's integrand $f=1/(x\,t)$ the residue is $y/x$. With $u=\bar g$, $x=(u-u^{-1})/2$ and $y=(u+u^{-1})/2$, so $y/x=(u^2+1)/(u^2-1)$ is not a Laurent polynomial in $u$, and the residue lies outside $\Fbar+\sum_j\Fbar\bar g^{\,j}$. Corollary~\ref{cor:rescert} upgrades Part~II's certificate from ``not elementary'' to ``not in $\GQ$, nor with elliptic integrals'':
\begin{verbatim}
>>> integrate_surface_special((x + sqrt(x**2 + 1))/(sqrt(x**2 + 1)
...                            *log(x + sqrt(x**2 + 1))), x)
(li(x + sqrt(x**2 + 1)), True)
>>> integrate_surface_special(1/(x*log(x + sqrt(x**2 + 1))), x)
('not in class', log(x + sqrt(x**2 + 1)), (0, 1/x),
 'canonical residue Fbar + span of the reductions of the exponential monomials
 constant at the prime (Corollary 4.3; relative to the sources
 ['y + x'])')
\end{verbatim}
The same moving prime, two integrands, and the residue decides between them. The residue is printed in the tower coordinates, $(0,1/x)\leftrightarrow y/x$.

\subsection{$\int e^x/(x^2+1)\,dx$: complex shifts}\label{ex:complexei}

The prime $x^2+1$ splits over $\Q(\sqrt{-1})$; at $x=\pm\sqrt{-1}$, $\hat\tau=\pm\bar\theta/(2\sqrt{-1})$ with $\bar\theta$ transcendental in $\kappa=\Q(\sqrt{-1})(\theta)$, and $c=\pm\sqrt{-1}$. The two kernels $\theta/(x\mp\sqrt{-1})$ are subtracted, the constants $e^{\pm\sqrt{-1}}$ appearing only at output:
\begin{verbatim}
>>> integrate_surface_special(exp(x)/(x**2 + 1), x)
(-I*exp(I)*Ei(x - I)/2 + I*exp(-I)*Ei(x + I)/2, True)
\end{verbatim}

\subsection{$\int e^{\sqrt x}/x\,dx$ and $\int e^x/\sqrt x\,dx$: flattened roots}\label{ex:expsqrtx}

The flattened tower of \cite[\S10.5]{PartII}: $L=\Q(u,\theta)$, $Du=1/(2u)$, $\theta=e^u$, $(u)$ normal with $\delta=2$. For $f=\theta/u^2$, $v_{(u)}(f)=-2=-\delta$ and $\hat\tau_{(u)}=2\bar\theta$: a residue at a $\delta=2$ prime, and the kernel $\theta\,Du/u$ with coefficient $2$. For $e^x/\sqrt x$ the flattened tower has $\theta=e^{u^2}$ and the Kummer search finds $\erfi$:
\begin{verbatim}
>>> integrate_surface_special(exp(sqrt(x))/x, x)
(2*Ei(sqrt(x)), True)
>>> integrate_surface_special(exp(x)/sqrt(x), x)
(sqrt(pi)*erfi(sqrt(x)), True)
\end{verbatim}

\subsection{$\int x\,e^{-x^2}/\sqrt{x^2+1}\,dx$: $\erf$ of the radical}\label{ex:conic}

$K_0=\Q(x,\theta)$, $\theta=e^{-x^2}$, $y^2=x^2+1$, the curve of the flagship examples of Part~II. $f=x\theta/y$ lies in the coordinate $y$ at weight $1$, $\Lambda=-x^2$. The shift $c=0$ gives the $\erf(x)$ kernel $-2\theta$, in the rational coordinate, which the system leaves with coefficient zero; the branch value $c=-1$ gives $v=-(x^2+1)=-y^2$, which vanishes at the branch prime $(y)$ to order $n=2=e_{(y)}$. So $s=\tfrac12$ is admissible by Corollary~\ref{cor:parameter} although $v$ has only simple zeros in $\Q[x]$, and $\omega=\theta/y$ passes the Kummer test ($W/q=\theta^2/(x^2+1)^2$). The kernel $-2x\theta/y$ has valuation $-1$ at the branch prime, where $\delta=2$: the window again. The session supplies the tower explicitly:
\begin{verbatim}
>>> q = x**2 + 1; T = Tower([x, theta], [(1, 0), (-2*x*theta, 0)], q)
>>> integrate_surface_special(x*exp(-x**2)/sqrt(q), x,
...     tower=(T, (0, theta*x/q), [(theta, exp(-x**2)), (Y, sqrt(q))]))
(-E*sqrt(pi)*erfc(sqrt(x**2 + 1))/2, True)
\end{verbatim}

\subsection{Hebisch's Example 22: the kernel is in $L$, its factors are not}\label{ex:hebisch22}

$K_0=\Q(x,\theta)$, $\theta=e^{-x}$, $y^2=q=2x\theta$, so that $y/x=e^{-x/2}/\sqrt{x/2}$; the radicand contains the exponential, and $(\theta)$ is a branch prime that is also special. $f=y/x$ is the integrand of \cite[Ex.~22]{Hebisch18}. The weight of $y$ is $\tfrac12$, so $J=\tfrac12$ and $\Lambda=-x/2$; with $c=0$, $v=-x/2$ vanishes at the branch prime over $(x)$ with $n=e=2$, so $s=\tfrac12$, and
\[
\omega=\bigl(\tfrac x2\bigr)^{-1/2}e^{-x/2}=\frac yx\in L,
\]
although neither $\sqrt x$ nor $e^{-x/2}$ lies in $L$. The trace of $f$ to $K_0$ is zero, so a reduction that passes to traces sees nothing; the parallel method works in the Trager coordinates and never takes a trace.
\begin{verbatim}
>>> q = 2*x*theta; T = Tower([x, theta], [(1, 0), (-theta, 0)], q)
>>> integrate_surface_special(sqrt(2*x*exp(-x))/x, x, strict=False,
...     tower=(T, (0, 1/x), [(theta, exp(-x)), (Y, sqrt(q))]))
(-2*sqrt(pi)*erfc(sqrt(2)*sqrt(x)/2), True)
\end{verbatim}

\subsection{$\Gamma(\tfrac13,\cdot)$}\label{ex:gamma13}

For $x^{-2/3}e^{-x}$ the root $y^3=x$ is flattened; the Kummer search finds a triple zero, $k\mid3$, and both $s=\tfrac13$ and $s=\tfrac23$ enter as columns, the system choosing the first. For $e^{-x^3}$ the triple zero is in $\Q[x]$:
\begin{verbatim}
>>> integrate_surface_special(x**Rational(-2, 3)*exp(-x), x)
(-uppergamma(1/3, x), True)
>>> integrate_surface_special(exp(-x**3), x)
(-uppergamma(1/3, x**3)/3, True)
\end{verbatim}

\subsection{$\Si$, $\Ci$ and a Fresnel integral}\label{ex:si}

For $\sin(x)/x$ the surface interface builds $L=\Q(x,t)$, $t=\tan(x/2)$, $f=2t/(x(1+t^2))$. At the normal prime $(x)$, $\hat\tau=2\bar t/(1+\bar t^2)$, with poles at $\bar t=\pm\sqrt{-1}$ only; with $\xi_t=(1+\sqrt{-1}t)/(1-\sqrt{-1}t)$, $\lambda(\xi_t)=\sqrt{-1}\,x$, one has $\hat\tau=(\bar\xi_t-\bar\xi_t^{-1})/(2\sqrt{-1})$, and the two $\Ei$ kernels $\xi_t^{\pm1}/x$ are subtracted. For $\sin(x^2)$ the Kummer search finds the kernels of $\Gamma(\tfrac12,\mp\sqrt{-1}x^2)$; the presentation returns them as complementary error functions at complex argument, which is where a Fresnel rewriting would enter.
\begin{verbatim}
>>> integrate_surface_special(sin(x)/x, x)
(Si(x), True)
>>> integrate_surface_special(cos(x)/x, x)
(Ci(x), True)
>>> integrate_surface_special(sin(x**2), x)
(-sqrt(I)*sqrt(pi)*erfc(sqrt(I)*x)/4
 - sqrt(pi)*sqrt(-I)*erfc(x*sqrt(-I))/4, True)
\end{verbatim}

\subsection{Silent zeros: $\int e^{e^x}dx$, $\int e^{1/x}dx$, $\int e^x\log x\,dx$}\label{ex:eiexp}\label{ex:explog}

$\int e^{e^x}\,dx$: $L=\Q(x,\theta_1,\theta_2)$, $\theta_1=e^x$, $\theta_2=e^{\theta_1}$; $f=\theta_2$ has no residues, and at the special prime $(\theta_1)$, $\lambda(\theta_2)=\theta_1\equiv0$, which gives the kernel $\theta_2\,D\theta_1/\theta_1=f$. For $e^{1/x}$ the argument $1/x$ vanishes at $x=\infty$; for $e^x\log x$ the argument $x$ vanishes at the special prime $(x)$ of $\log x$, and the $\Ei$ column enters together with the rational part $e^x\log x$:
\begin{verbatim}
>>> integrate_surface_special(exp(exp(x)), x)
(Ei(exp(x)), True)
>>> integrate_surface_special(exp(1/x), x)
(x*exp(1/x) - Ei(1/x), True)
>>> integrate_surface_special(exp(x)*log(x), x)
(exp(x)*log(x) - Ei(x), True)
\end{verbatim}

\subsection{The Lambert function}\label{ex:lambert}

$L=\Q(x,W)$, $DW=W/(x(1+W))$, the non-monomial tower of \cite[Ex.~10.3.4]{Bronstein05} in which $(W)$ is special. The source $(x/W,W)$ has $\lambda\equiv0$ at $(W)$, and the silent kernel $(x/W)\,DW/W=1/(W(1+W))$ is the integrand. No certificate is available (Section~\ref{sec:necessity}), so the session passes \texttt{strict=False}:
\begin{verbatim}
>>> integrate_surface_special(1/(LambertW(x)*(1 + LambertW(x))), x,
...                           strict=False)
(Ei(LambertW(x)), True)
\end{verbatim}

\subsection{A Hermite part with $\erfc$}\label{ex:hermite}

$\int e^{-x^2}/x^2\,dx$ needs both the Hermite denominator $x$ and the $\erf$ column; the bounds of Corollary~\ref{cor:bounds} are taken over the integrand and the kernel:
\begin{verbatim}
>>> integrate_surface_special(exp(-x**2)/x**2, x)
(sqrt(pi)*erfc(x) - exp(-x**2)/x, True)
\end{verbatim}

\subsection{Certificates, and a failure}\label{ex:certificate}

For $\log(x)/(x+1)$ and $\arctan(x)/x$ the residue at $(x+1)$, respectively $(x)$, is the reduction $\bar t$ of the logarithmic generator, transcendental in $\kappa$ and not in the exponential span; for $1/(x+e^x)$ the residue at the moving prime $x+\theta$ is $-1/(x-1)$. All three are certified outside the class. The first two have antiderivatives in dilogarithms (Section~\ref{sec:conclusion}). For $\log(x)\,e^{-x^2}$ the only candidate is the $\erf$ kernel, and in the part of weight one and degree one in $\log x$ the system demands a polynomial $b_1$ with $b_1'-2xb_1=1$, which does not exist. Every bound in force is proved --- the degree in $x$ by (K5) at $x=\infty$ (Example~\ref{ex:boundsdecided}(i)), where Part~II's criteria leave a guess --- so the inconsistency is not an artefact of the ansatz. Since special functions enter \eqref{eq:liouG} linearly with constant coefficients, $\tfrac{\sqrt\pi}2\erf(x)\log x-\tfrac{\sqrt\pi}2\int\erf(x)\,dx/x$ is not of the form \eqref{eq:liouG}, and no antiderivative in $\GQ$ is expected; the algorithm returns ``failed'' with the bounds recorded as proved, the one missing ingredient of a certificate being the completeness of the candidate lists.
\begin{verbatim}
>>> integrate_surface_special(log(x)/(x + 1), x)[:3]
('not in class', x + 1, (t1, 0))
>>> integrate_surface_special(1/(x + exp(x)), x)[:3]
('not in class', x + exp(x), (-1/(x - 1), 0))
>>> integrate_surface_special(log(x)*exp(-x**2), x)[:2]
('failed',
 'no solution with the special kernels (every bound in force proved)')
\end{verbatim}

\subsection{$\int dx/\sqrt{e^x-e^{-x}}$: a curve carried by an exponential}\label{ex:sinh}

$L=\Q(x,\theta,y)$, $\theta=e^x$, $y^2=\theta^3-\theta$, so $y/\theta=\sqrt{e^x-e^{-x}}$ and the integrand is $\theta/y$. The pencil is $u=\theta$, $\cC:Y^2=U^3-U$, $e_1=1$, $e_2=0$, $e_3=-1$, $\lambda^2=2$, $m=\tfrac12$. By Proposition~\ref{prop:points} the curve can only enter through the radical; here the radicand is a polynomial in the exponential, and the branch prime $(\theta)$ is also the special prime of the hyperexponential. The columns are $D\theta/y=\theta/y$ and $\theta\,D\theta/y=\theta^2/y$; the system gives $f=K^{(1)}$, and with $\sin\varphi=\sqrt{2/(\theta+1)}$:
\begin{verbatim}
>>> q = theta**3 - theta; T = Tower([x, theta], [(1, 0), (theta, 0)], q)
>>> tw = (T, (0, theta/q), [(theta, exp(x)), (Y, sqrt(q))])
>>> integrate_surface_special(exp(x)/sqrt(exp(3*x) - exp(x)), x, tower=tw)
(-sqrt(2)*elliptic_f(asin(sqrt(2)/sqrt(exp(x) + 1)), 1/2), True)
\end{verbatim}

\subsection{$\int\sqrt{e^x-e^{-x}}\,dx$: the second kind and a Hermite part}\label{ex:sinh2}

Same tower, $f=y/\theta$. The system $f=D(b)+\alpha K^{(1)}+\beta K^{(2)}$ has the solution $b=-2y/\theta$, $\alpha=0$, $\beta=2$ (indeed $D(y/\theta)=(\theta^2+1)/(2y)$), the denominator $\theta$ of $b$ sitting at the special branch prime, where the bound uses the valuation of $K^{(2)}$:
\begin{verbatim}
>>> tw = (T, (0, 1/theta), [(theta, exp(x)), (Y, sqrt(q))])
>>> integrate_surface_special(sqrt(exp(3*x) - exp(x))/exp(x), x, tower=tw)
(-2*sqrt(exp(3*x) - exp(x))*exp(-x)
 + 4*sqrt(2)*elliptic_e(asin(sqrt(2)/sqrt(exp(x) + 1)), 1/2)
 - 2*sqrt(2)*elliptic_f(asin(sqrt(2)/sqrt(exp(x) + 1)), 1/2)
 + 4*sqrt(exp(3*x) - exp(x))/(exp(x) + 1), True)
\end{verbatim}

\subsection{$\int dx/\bigl((x-2)\sqrt{x^3-x}\bigr)$: a non-torsion residue becomes $\Pi$}\label{ex:third}

$L=\Q(x,y)$, $y^2=x^3-x$, $f=1/((x-2)y)$, with residues $\pm1/\sqrt6$ at the places $(2,\pm\sqrt6)$. Reduction modulo primes in which $6$ is a square gives the orders $4$ modulo $5$ and $3$ modulo $23$ for the place $(2,\sqrt6)$, which are incompatible with finite order by \cite[Prop.~9.4]{PartII}; Part~II therefore certifies non-elementarity. Proposition~\ref{prop:third} gives $c=1$ and $K_{x-2}=f$, and Lemma~\ref{lem:legendre} with $n=3/2$:
\begin{verbatim}
>>> q = x**3 - x; T = Tower([x], [(1, 0)], q)
>>> integrate_surface_special(1/((x - 2)*sqrt(q)), x,
...     tower=(T, (0, 1/((x - 2)*q)), [(Y, sqrt(q))]))
(sqrt(2)*elliptic_f(asin(sqrt(2)/sqrt(x + 1)), 1/2)/3
 - sqrt(2)*elliptic_pi(3/2, asin(sqrt(2)/sqrt(x + 1)), 1/2)/3, True)
\end{verbatim}

\subsection{Elliptic objects in the integrand}\label{ex:ellK}\label{ex:sn2}

$K(m)$ and $E(m)$ are non-monomial generators, $DK=(E-(1-m)K)/(2m(1-m))$, $DE=(E-K)/(2m)$, and Part~II computes $\int K(m)\,dm$ with no special machinery. For $s=\operatorname{sn}(x\,|\,k)$, $y=\operatorname{cn}\operatorname{dn}$, $y^2=(1-s^2)(1-ks^2)$, $Ds=y$, the rescaling $D'=y^{-1}D$ \cite[Lemma~3.4]{PartII} gives a pencil in $s$ with $D's=1$ above which the variable is an upper primitive with $D'x=1/y$; the integrand becomes $s^2/y$, and the Legendre case of Lemma~\ref{lem:legendre} gives $\int\operatorname{sn}^2(x\,|\,k)\,dx=(F-E)(\operatorname{am}x\,|\,k)/k=(x-E(\operatorname{am}x\,|\,k))/k$:
\begin{verbatim}
>>> m = symbols('m', positive=True); K, E = symbols('K E')
>>> T = Tower([m, K, E], [(1, 0), ((E - (1 - m)*K)/(2*m*(1 - m)), 0),
...                       ((E - K)/(2*m), 0)])
>>> parallel_integrate_special((K, 0), T)
2*E + 2*K*m - 2*K
>>> s, xs = symbols('s x_s'); q = expand((1 - s**2)*(1 - s**2/3))
>>> T = Tower([s, xs], [(1, 0), (0, 1/q)], q)
>>> parallel_integrate_special((0, s**2/q), T).expr
3*x_s - 3*elliptic_e(asin(s), 1/3)
\end{verbatim}
The first-kind integral $F(\arcsin s\,|\,\tfrac13)$ is the generator $x_s$ itself, and the presentation writes it so (Section~\ref{sec:noneed}).

\subsection{Third-kind residue divisors over number fields}\label{ex:numberfield}

For $f=1/\bigl((x^2-5)\sqrt{x^3-x}\bigr)$ the prime $x^2-5$ is irreducible over $\Q$. The four places over it are
\[
P_{\pm}=\bigl(\sqrt5,\pm2\cdot5^{1/4}\bigr),\qquad Q_{\pm}=\bigl(-\sqrt5,\pm2\sqrt{-1}\,5^{1/4}\bigr),
\]
with residues $\pm5^{1/4}/20$ and $\pm\sqrt{-1}\,5^{1/4}/20$. Their $\Q$-span has the basis $r_1=5^{1/4}/20$, $r_2=\sqrt{-1}\,r_1$, and the two components $P_+-P_-$ and $Q_+-Q_-$ are Galois conjugate: on $y^2=x^3-x$, which has complex multiplication by $\sqrt{-1}$, in fact $Q_\pm=[\sqrt{-1}]P_\pm$. Proposition~9.4 of Part~II applies at degree-one primes of $K=\Q(5^{1/4},\sqrt{-1})$, and three primes are needed: the order of $P_+-P_-$ is $13$ modulo $101$, $13$ modulo $109$ and $17$ modulo $149$. The first two are compatible with a point of order $13$; the third excludes it. The implementation had stopped at two primes and reported ``needs torsion realisation''. It now samples up to eight, certifies the divisor non-torsion, and the third-kind kernel of Proposition~\ref{prop:third} gives $\Pi$ at the characteristics $(1\pm\sqrt5)/2$:
\begin{verbatim}
>>> q = x**3 - x; T = Tower([x], [(1, 0)], q)
>>> integrate_surface_special(1/((x**2 - 5)*sqrt(q)), x,
...     tower=(T, (0, 1/((x**2 - 5)*q)), [(Y, sqrt(q))]))
(sqrt(10)*elliptic_f(phi, 1/2)/(10 + 10*sqrt(5))
 - sqrt(10)*elliptic_f(phi, 1/2)/(10 - 10*sqrt(5))
 + sqrt(10)*elliptic_pi(1/2 - sqrt(5)/2, phi, 1/2)/(10 - 10*sqrt(5))
 - sqrt(10)*elliptic_pi(1/2 + sqrt(5)/2, phi, 1/2)/(10 + 10*sqrt(5)),
 True)                          # phi = asin(sqrt(2)/sqrt(x + 1))
\end{verbatim}
The same machinery separates the torsion cases from the others:
\begin{itemize}
\item For $x^2-2x-1$ the places lie over $1\pm\sqrt2$, the $x$-coordinates of points $P$ with $2P=(1,0)$, which are $4$-torsion. Part~II realises the logarithm over $\Q(\sqrt2)$, and the answer is that logarithm plus a first-kind term.
\item For $x^2+1$ the points over $\pm\sqrt{-1}$ are $4$-torsion as well, since $2P=(0,0)$.
\item For $x^2-3x+1$ the divisor is non-torsion (orders $5$ modulo $401$, $58$ modulo $449$), and the answer is in $\Pi$.
\end{itemize}
All four are certified non-elementary and independent under Theorems~\ref{thm:necessity} and \ref{thm:independence}.

\subsection{Cherry's classes}\label{ex:cherry}

Cherry's algorithms decide, for an integrand in a transcendental elementary tower, whether it has an antiderivative in elementary functions and error functions \cite{Cherry85}, respectively elementary functions and logarithmic integrals \cite{Cherry86}. Both are recursive: they solve Risch differential equations, with parameters. Their classes lie in $\GQ$ --- $\erf$ is $\Gamma(\tfrac12,\cdot)$, and $\li(u)=\Ei(\log u)$ with $\log u$ in the tower --- and on transcendental elementary towers the hypotheses (H1) and (H2) hold by \cite[Thm.~18]{Hebisch18} and \cite{RothsteinCaviness79}.

We tested twenty integrals of Cherry's types:
\begin{itemize}
\item $e^{-x^2}$, $x^2e^{-x^2}$, $e^{-x^2+x}$, $e^xe^{-e^{2x}}$, $(1+x^{-2})e^{x^2}$, and Gaussians of $\log x$ such as $e^{-\log^2x}$ and $x^2e^{-\log^2x}$;
\item $1/\log(x+1)$, $x^3/\log x$, $x/(1+\log x)$, $1/(\log^2x-1)$ and other sums of logarithmic integrals;
\item one elementary integral and three outside the class.
\end{itemize}
Algorithm~S7 finds every antiderivative, including $\int e^{-\log^2x}\,dx=-\tfrac{\sqrt\pi}2e^{1/4}\erfc(\log x-\tfrac12)$. That answer needed two additions made for this group: the Kummer search now combines the arguments of logarithmic generators with the exponentials, and criterion (K8) decides the places at $x=\infty$ and $(x)$ of the Gaussian over $\log x$. It certifies the necessity of the special terms in sixteen of the twenty cases.

The four Gaussians of $\log x$ remain uncertified. In $\Q(x,t,e^{-t^2})$, $t=\log x$, the place $t=\infty$ is a chain place in the sense of Part~II: the Gaussian depends on the primitive $t$. There Theorem~\ref{thm:extra} gives $r^*=1<r=2$ and $\partial^*=d/d\bar x$, so the constant leading forms defeat (ii*). This is precisely where Cherry solves a Risch differential equation.

Two cases decide the question the other way. For $\log(x)\,e^{-x^2}$ Cherry's algorithm proves that no antiderivative in his class exists. Algorithm~S7 proves every bound and finds the system inconsistent, but reports ``failed'', because the completeness of its candidate lists is not proved. For $1/(x+\log x)$ both decide: the residue certificate proves non-membership.

So the extended algorithm covers Cherry's classes as \emph{answers}, together with $\Gamma(s,\cdot)$ at every rational $s$, $\Si$, $\Ci$, Fresnel integrals, elliptic integrals, radicals anywhere in the tower and non-monomial generators, with erf and li terms in one answer. It does not yet subsume them as \emph{decision procedures}; Section~\ref{sec:conclusion} lists what that requires.

\section{Conclusions and open problems}\label{sec:conclusion}

Special functions in the answer fit the parallel method without changing its shape. A kernel is a known element of the tower; its coefficient is a constant of the one linear system, or is fixed by the residues before the system is assembled; and the valuation theory of Part~II decides where kernels can have poles. Three facts organise the extension: the $s=0$ kernels are exactly what a non-constant residue in the exponential span asks for; the $s>0$ kernels live in the sub-critical window of radical towers; and elliptic integrals live on the radical, their third kind taking over Part~II's non-torsion certificate. The residue field carries a hierarchy: constant residues are the residues of logarithms and third-kind integrals, residues in the span of reductions of exponential elements those of $\Ei$, and residues affine in the reductions of logarithmic generators --- $\log(x)/(x+1)$ and $\arctan(x)/x$ of \S\ref{ex:certificate} --- those of dilogarithms.

The open problems are the hypotheses and the completeness of the candidate lists.
\begin{enumerate}
\item[(a)] \emph{Descent through the radical} (H1), for the kernel data of Theorem~\ref{thm:liouville}; the pitfall is Hebisch's Example~22, where a trace of $f$ vanishes.
\item[(b)] \emph{The exponential group} (H2): generators of $\fE/\Fbar^*$ for towers (T1)--(T4), including the Kummer closure of the radical.
\item[(c)] \emph{Completeness of the candidates}: that the residue rule and the silent-zero rule produce every $\Ei$ term, and the Kummer search every $\Gamma(s,\cdot)$ term with $s>0$, up to kernels whose sum is elementary; and the independence of distinct kernels modulo elementary functions, which Theorem~\ref{thm:independence} certifies instance by instance and which, in general, would make the minimality of Proposition~\ref{prop:elementaryfirst}(E4) absolute. With (a)--(c) an inconsistent final system with proved bounds becomes a certificate of non-membership, as the holomorphic remainder of Part~II is for elementary functions; \S\ref{ex:certificate}, where every bound is now proved, is the test case, and (c) is the only missing ingredient there.
\item[(d)] \emph{Pencils}: that every non-constant point of a constant elliptic curve over $\Fbar K_0(y)$ factors through a pencil, and an algorithm for the multivariate case; the implementation detects pencils in one generator.
\item[(e)] \emph{The dilogarithm}: we conjecture that adjoining the logarithms of the primes at which the residue is affine in a logarithmic generator, and allowing a rational part of degree two in them, reduces the dilogarithm terms to solutions of the unit equation $u+(1-u)=1$ in the group generated by those primes and the arguments of the logarithms --- an $S$-unit problem of the kind Parts~I--II already solve. For $\log(x)/(x+1)$: $T=\log(x+1)$, $u=-x$, and $\int=\log x\log(1+x)+\Li_2(-x)$.
\item[(f)] \emph{Bounds at the zeros of exponential elements}: by Proposition~\ref{prop:kerexp} the places still left to the guess are, in our towers, those where an exponential element is constant to first order; deciding them is a Risch equation in the residue field, which the series declines to solve, and the reduction systems of \cite{DuRaab25} are the natural tool should one wish to.
\item[(g)] \emph{Deciding torsion over number fields.} The non-torsion search terminates: for a non-torsion class the reduction orders at degree-one primes are unbounded, and Chebotarev supplies the primes. The torsion side should use the order $N_0$ that the compatible reductions predict, rather than the fixed search bound of Part~II's realisation, and then test $N_0\Delta$ for principality over $K$. Together the two would make the residue analysis a decision procedure for every residue divisor on a genus-one pencil, as Risch's and Trager's theory requires.
\item[(h)] \emph{Several third-kind kernels.} Theorem~\ref{thm:independence}(b) allows one third-kind kernel. Several need a joint certificate: that no nontrivial $\Z$-combination of the component classes is torsion, which the relation lattices of the reductions can certify.
\item[(i)] \emph{Algebraic coefficients.} An integrand with algebraic coefficients, such as $1/\bigl((x-1-\sqrt2)\sqrt{x^3-x}\bigr)$, is correct in principle but exceeds $100$ seconds in Part~II's arithmetic over $\Q(\sqrt2)$.
\item[(j)] \emph{Cherry's decision procedures.} On transcendental elementary towers, where (H1) and (H2) hold, two things would make Algorithm~S7 a decision procedure for Cherry's classes and for $\GQ$. The first is (c): the completeness of the residue rule and the silent-zero rule for $\Ei$ and $\li$, and of the Kummer search, with the joint monomials of \S\ref{ex:cherry}, for $\erf$. The second is proved bounds at the chain places of the Gaussians, a leading-part analysis of higher order than Theorem~\ref{thm:extra} at $t=\infty$ in $\Q(x,t,e^{-t^2})$. The final inconsistency of the linear system would then be a certificate of non-membership, and $\int\log(x)e^{-x^2}\,dx$ would be decided rather than failed.
\item[(k)] \emph{Presentation}: Fresnel forms for the $\Gamma(\tfrac12,\sqrt{-1}\,\cdot)$ pairs, Legendre forms on quartic and genus-$\ge2$ pencils with elliptic factors, and Jacobi's theta form of the third kind.
\end{enumerate}

\section*{Data availability}
The SymPy implementation used for every session in Section~\ref{sec:examples} (\texttt{parallel\_mixed\_special.py}, built on the implementation of Part~II --- \texttt{parallel\_mixed.py}, \texttt{logrewrite.py}, \texttt{pell.py} and \texttt{build\_tower.py}, the last with the structure-theorem test of \S\ref{sec:functions} --- together with \texttt{examples\_special.py}, which reruns every example and reports the outcome claimed for it, \texttt{bounds\_special.py}, which reproduces the statistics of Section~\ref{sec:bounds}, \texttt{elementary\_check.py}, which reruns the tests of Section~\ref{sec:noneed}, \texttt{partial\_check.py}, which reruns the test of Section~\ref{sec:partial}, and \texttt{stress\_special.py}, a corpus of $251$ integrals labelled with their mathematically expected outcomes) is provided as supplementary material.

\end{document}